\documentclass[a4paper]{article}
\usepackage{a4wide}
\usepackage{authblk}
\makeatletter
\providecommand*{\shuffle}{%
  \mathbin{\mathpalette\shuffle@{}}%
}
\newcommand*{\shuffle@}[2]{%
  \sbox0{$#1\vcenter{}$}%
  \kern .15\ht0 
  \rlap{\vrule height .25\ht0 depth 0pt width 2.5\ht0}%
  \raise.1\ht0\hbox to 2.5\ht0{%
    \vrule height 1.75\ht0 depth -.1\ht0 width .17\ht0 %
    \hfill
    \vrule height 1.75\ht0 depth -.1\ht0 width .17\ht0 %
    \hfill
    \vrule height 1.75\ht0 depth -.1\ht0 width .17\ht0 %
  }%
  \kern .15\ht0 
}
\makeatother

\usepackage{graphicx}
\usepackage{tikz}
\usetikzlibrary{fit,patterns,decorations.pathmorphing,decorations.pathreplacing,calc,arrows}

\usepackage{amsmath,amssymb,amsthm}
\usepackage[inline]{enumitem}

\usepackage{thmtools}
\usepackage{hyperref}
\usepackage[capitalize,nameinlink,sort]{cleveref}

\usepackage{mleftright}
\mleftright

\usepackage[T1]{fontenc}
\usepackage{newpxtext}
\usepackage{euler}

\usepackage[textsize=scriptsize]{todonotes}

\DeclareMathOperator{\Pow}{Pow}
\DeclareMathOperator{\N}{\mathbb{N}}

\DeclareMathOperator{\Z}{\mathbb{Z}}

\newcommand{\qbin}[3][q]{\binom{#2}{#3}_{\!\! #1}}

\declaretheorem[numberwithin=section]{theorem}
\declaretheorem[sibling=theorem]{lemma,corollary,proposition}
\declaretheorem[sibling=theorem,style=definition]{example,definition,remark}

\declaretheorem{claim}

\declaretheoremstyle[
    headfont=\normalfont\itshape, 
    bodyfont = \normalfont,
    qed=$\blacksquare$, 
    headpunct={:}]{claimproofstyle} 
\declaretheorem[name={Proof of claim}, style=claimproofstyle, unnumbered]{claimproof}

\crefname{equation}{}{}

\title{$q$-Parikh Matrices and $q$-deformed binomial coefficients of words}

\author{Antoine Renard}

\author{Michel Rigo\thanks{ORCID 0000-0001-7463-8507}}

\author{Markus A. Whiteland\thanks{Supported by the FNRS Research grant 1.B.466.21F, ORCID 0000-0002-6006-9902.}}

\affil{Department of Mathematics, University of Li\`ege, Li\`ege, Belgium}

\affil{\texttt{\{antoine.renard,m.rigo,mwhiteland\}@uliege.be}}

\date{}

\begin{document}
\maketitle

\begin{abstract}
 We have introduced a $q$-deformation, i.e., a polynomial in $q$ with natural coefficients, of the binomial coefficient of two finite words $u$ and $v$ counting the number of occurrences of $v$ as a subword of $u$.
 In this paper, we examine the $q$-deformation of Parikh matrices as introduced by E{\u{g}}ecio{\u{g}}lu in 2004. 

 Many classical results concerning Parikh matrices generalize to this new framework: Our first important observation is that the elements of such a matrix are in fact $q$-deformations of binomial coefficients of words.
 We also study their inverses and as an application, we obtain new identities about $q$-binomials.

 For a finite word $z$ and for the sequence $(p_n)_{n\ge 0}$ of prefixes of an infinite word, we show that the polynomial sequence $\qbin{p_n}{z}$ converges to a formal series.
 We present links with additive number theory and $k$-regular sequences. In the case of a periodic word $u^\omega$, we generalize a result of Salomaa: the sequence $\qbin{u^n}{z}$ satisfies a linear recurrence relation with polynomial coefficients.
Related to the theory of integer partition, we describe the growth and the zero set of the coefficients of the series associated with $u^\omega$.

Finally, we show that the minors of a $q$-Parikh matrix are polynomials with natural coefficients and consider a generalization of Cauchy's inequality. We also compare $q$-Parikh matrices associated with an arbitrary word with those associated with a canonical word $12\cdots k$ made of pairwise distinct symbols.
\end{abstract}

\section{Introduction}

In a recent paper \cite{RRW}, we have introduced $q$-deformations of binomial coefficients of words.
\begin{definition}\label{def:qbin}
  The {\em $q$-binomial coefficient} of two words over a finite alphabet $A$ is a polynomial of $\N[q]$ recursively defined as follows. For all words $u$,
  $v \in A^*$ and letters $a,b \in A$, we consider a $q$-deformation of Pascal's identity:
  \begin{equation}
    \label{eq:recdef}
    \qbin{u}{\varepsilon} = 1, \quad \qbin{\varepsilon}{v} = 0 \text{ if }v \neq \varepsilon, \quad \text{and } \quad \qbin{ua}{vb} = \qbin{u}{vb}\cdot q^{|vb|} + \delta_{a,b}\qbin{u}{v}.
  \end{equation}
\end{definition}
As usual when dealing with $q$-deformations of a combinatorial quantity, letting $q$ tend to~$1$ gives back the classical binomial coefficient $\binom{u}{v}\in\mathbb{N}$ of the words $u$ and $v$ counting the number of occurrences of $v$ as a subword of $u$. See, for instance, \cite{Lothaire1997} for a survey. As an example, for $u=u_1\cdots u_6=abaaba$ over the alphabet $\{a,b\}$, we have
\[\qbin{u}{ba}=q^6+q^5+q^3+1\text{ and }\binom{u}{ba}=4\]
because $u_2u_3=u_2u_4=u_2u_6=u_5u_6=ba$ and there are no other occurrence of $ba$. It is quite evident that the polynomial $q^6+q^5+q^3+1$ contains more information about the occurrences of $ba$ in $u$ than its evaluation at $1$, which tells us only about the appearance of four subwords $ba$. Indeed, for example the constant term $1$ implies that $ba$ is a suffix of $u$.

\subsection{Parikh matrices}
Let $k$ be the size of the alphabet~$A$. For convenience (in particular, for indexing matrix entries), we identify $A$ with $\{1,\ldots,k\}$. 
When considering binomial coefficients of words, it is a quite natural question to look at the so-called Parikh matrices introduced more than twenty years ago by Mateescu {\em et al.} \cite{MateescuSalomaa2001}. Associated with a finite word~$u\in A^*$, it is an upper triangular matrix $M(u)$ of size $(k+1)\times (k+1)$ whose elements are binomial coefficients $\binom{u}{v}$ for words $v$ of the form $i\, (i+1)\cdots j$, $1\le i\le j\le k$. In particular, on the second diagonal $(M_{i,i+1})$, one finds the $k$~coefficients~$\binom{u}{d}$ for the letters $d$ in the alphabet~$A$. So the second diagonal encodes what is called the {\em Parikh vector} or {\em abelianization} of $u$.

There is a vast literature on Parikh matrices. See, for instance, \cite{MateescuSalomaa2001,MateescuSalomaa2004,SalomaaCauchy2003}. A particular question that has attracted a lot of attention is the {\em $M$-ambiguity problem} or, injectivity problem, that is: characterize words with the same Parikh matrix \cite{FosseRichomme,SalomaaCriteria2010,WenSubramanian2018}. It can be related to the famous {\em reconstruction problem}: does the knowledge of binomial coefficients $\binom{u}{v}$ for some words~$v$ permit to uniquely determine the word~$u$? See, for instance, \cite{Fleischmann2021,richomme2023reconstructing}.

\c{S}erb\u{a}nu\c{t}\u{a} has generalized Parikh matrices to matrices induced by a word~$z=z_1\cdots z_n$~\cite{Serbanuta2004}. These matrices have size $(|z|+1)\times (|z|+1)$ and they contain elements of the form $\binom{u}{v}$ for words~$v$ of the form $z_i\, z_{i+1}\cdots z_j$, $1\le i\le j\le |z|$. Taking $z=12\cdots k$ leads back to the definition given by Mateescu {\em et al}.

\subsection{$q$-Parikh matrices}
Having at our disposal our own variation \cref{eq:recdef} of binomial coefficients of words, we then wanted to define the corresponding $q$-deformation of Parikh matrices. Several authors have already tackled this question by immediately considering a convenient $q$-deformation of Parikh matrices. This approach has led to a $q$-analogue of the $M$-ambiguity problem. However it was not linked nor aimed to provide any $q$-analogue to the binomial coefficients of words. See, for instance, \cite{BeraMahalingam2016,Egecioglu}


On the other hand, prior to \cite{Egecioglu}, E{\u{g}}ecio{\u{g}}lu has introduced, for $d\in A=\{1,\ldots,k\}$ and $j\ge 0$, an infinite family of $(k+1)\times(k+1)$ matrices $\mathcal{M}_{d,j}$ with $1$'s on the diagonal and the only other non-zero element is $q^j$ in position $(d,d+1)$ on the second diagonal  \cite{Eeciolu2004}. Multiplying together these matrices gives what we consider to be the appropriate way of introducing $q$-deformations of Parikh matrices. For instance, with $A=\{1,2,3\}$, the matrices $\mathcal{M}_{1,j}$, $\mathcal{M}_{2,j}$ and $\mathcal{M}_{3,j}$ are
\begin{table}[h!tb]
  \[
{\small\left(
\begin{array}{cccccc}
 1 & q^j & 0 & 0 \\
 0 & 1 & 0 & 0  \\
 0 & 0 & 1 & 0  \\
 0 & 0 & 0 & 1  \\
\end{array}
\right)},\ {\small\left(
\begin{array}{cccccc}
 1 & 0 & 0 & 0 \\
 0 & 1 & q^j & 0  \\
 0 & 0 & 1 & 0  \\
 0 & 0 & 0 & 1  \\
\end{array}
\right)},\ {\small\left(
\begin{array}{cccccc}
 1 & 0 & 0 & 0 \\
 0 & 1 & 0 & 0  \\
 0 & 0 & 1 & q^j  \\
 0 & 0 & 0 & 1  \\
\end{array}
\right)}. 
\]
  \caption{Three E{\u{g}}ecio{\u{g}}lu's matrices for $A=\{1,2,3\}$.}
  \label{tab:ege}
\end{table}

Indeed, with his definition which differs from the one found in \cite{Egecioglu}, the entries of the upper triangular matrices are $q$-binomial coefficients $\qbin{u}{i\cdots j}$ up to multiplication by some well-understood powers of $q$ depending on the position within the matrix, for $1\le i\le j\le k$. They can therefore be linked to \cref{def:qbin} and our combinatorial object of interest.

The present paper is thus the occasion to shed some new light on E{\u{g}}ecio{\u{g}}lu's work \cite{Eeciolu2004}. We believe that a reason why the latter research work has not been much explored could be that, with such a definition, the problem of $M$-ambiguity does not arise. Indeed, E{\u{g}}ecio{\u{g}}lu has already noticed that two distinct words have different associated matrices: it is an immediate result on $q$-binomials that the knowledge of $\qbin{u}{d}$ provides the exact positions of the letter~$d$ within $u$ (see \cref{thm:powers}). Consequently, such a $q$-deformation does not provide any hint to the classical injectivity problem of Parikh matrices. Our aim is thus to explore the connections with our $q$-deformation of binomial coefficients of words.

\subsection{Our contributions}
In \cref{sec:2}, similarly to \c{S}erb\u{a}nu\c{t}\u{a} we extend E{\u{g}}ecio{\u{g}}lu's definition to matrices induced by a word~$z$.
We first use basic linear algebra techniques to obtain direct analogues of classical results on Parikh matrices: 
we discuss the significance of the entries of these matrices and their inverse with respect to $q$-binomial coefficients of words. Up to a multiplication by a convenient power of~$q$, the elements in the upper part of such a matrix are of the form $\qbin{u}{z_i\cdots z_j}$ for $1\le i\le j\le |z|$.

For results about inverses, we work under the assumption that the fixed word~$z$ has no factor $aa$  made of two identical letters. With \cref{rem:no_aa} we see why the inverse has a more intricate form when $z$ does not fulfill this assumption. We also get that the inverse of a $q$-Parikh matrix associated with~$z$ is directly related (up to some reverse operation and sign alternation) to the $q$-Parikh matrix associated with the reversal~$\widetilde{z}$.

As an application of our first results on $q$-Parikh matrices and their inverses, we obtain new relations on those $q$-binomials. For instance, we obtain identities such as
\[
\qbin{u}{a}\qbin{u}{cb}+q^2\qbin{u}{abc}=
  \qbin{u}{c}\qbin{u}{ab}+q^2\qbin{u}{cba},
  \]
  where $u$ is a word and $a,b,c$ are letters such that $a\neq b$ and $b\neq c$. \cref{pro:general_relation2} generalizes this relation for more than three letters with the restriction that adjacent letters are distinct.

\cref{sec:convergence} takes a different perspective, this time introducing new objects of an algebraic nature. It is quite evident that with any left-infinite word~$\mathbf{x}$, if $p_n$ is the prefix of length $n$ of~$\mathbf{x}$, then the sequence $n\mapsto\qbin{p_n}{z}$ converges to a formal power series $\mathfrak{s}_{\mathbf{x},z}$ in $\mathbb{N}[[q]]$. If the infinite word~$\mathbf{x}$ is $k$-automatic, then $\mathfrak{s}_{\mathbf{x},z}$ is shown to be $k$-regular. In particular, our considerations make it possible to express certain sequences ({\tt A133009} from OEIS) encountered in additive number theory, respectively, on the number of representations of an integer $n$ as the sum of two \emph{odious} or \emph{evil} numbers \cite{allouche2022additive,ChenWang,Dombi,ErdosIV}, in the framework of $q$-binomial coefficients.

Let $u,z$ be finite words. Salomaa has shown that the integer sequence $n\mapsto\binom{u^n}{z}$ satisfies a linear recurrence relation with constant integer coefficients \cite{Salomaa2008}. For this specific case, it is therefore natural to ask what more can be said about the sequence of polynomials $n\mapsto\qbin{u^n}{z}$. As a generalization of Salomaa's result, we show that it satisfies a linear recurrence relation over
$\N[q]$.

Finally for a periodic infinite word $\mathbf{u}=\cdots uuu$, we have a precise description of the growth order of the coefficients of the series $\mathfrak{s}_{\mathbf{u},z}(q)$. It will be observed by classical arguments that these coefficients also satisfy a linear recurrence relation with integer coefficients. We show that $n\mapsto [q^n]\mathfrak{s}_{\mathbf{u},z} $ vanishes periodically and those conditions are prescribed by a well understood arithmetic relation. If it does not vanish then it is in $\Theta(n^{|z|-1})$. Interestingly, this study also makes connection with the extensively studied problems of partition of integers into distinct parts \cite{Knessl1990}. 

In the last section, with \cref{cor:new_expression} we express a $q$-Parikh matrix associated with an arbitrary word~$z$ as a $q$-Parikh matrix (of the same dimension) associated with a canonical word of the form $12\cdots |z|$. To that end, we introduce a morphism $\sigma_z$ encoding the positions of a letter occurring in~$z$. In particular, such a word $12\cdots |z|$ satisfies the assumption of \cref{thm:inverse2} which permits us to also provide an expression for the inverse. 
Then we consider minors of $q$-Parikh matrices. Interestingly, they are polynomials with non-negative integer coefficients. This is an alternative explanation of the fact that the inverse of a $q$-deformed Parikh matrix has entries in $\pm\mathbb{N}[q]$. Finally, we discuss a $q$-analogue of what Salomaa calls the Cauchy inequality \cite{SalomaaCauchy2003}: for all words $x,y,z,w$, the polynomial
  \[\qbin{xy}{w}\qbin{yz}{w}-\qbin{xyz}{w}\qbin{y}{w}\]
  has non-negative integer coefficients.

  All these results tend to show that our definition of a $q$-deformation of Parikh matrices is the right one, since we generalize results from several different papers \cite{Serbanuta2004,MateescuSalomaa2001,SalomaaCauchy2003,Salomaa2008}.
  
\section{$q$-deformation of Parikh matrices}\label{sec:2}

We generalize E{\u{g}}ecio{\u{g}}lu's definition using \c{S}erb\u{a}nu\c{t}\u{a}'s approach as described in the introduction. We get matrices associated with a word~$z$ and whose entries are polynomials in~$q$.

\begin{definition}\label{def:ourdef}
   Let $z=z_1\cdots z_\ell$ be a word and $A$ be the alphabet of $z$, i.e., the set of letters occurring in $z$.
   For $d\in A$ and $j\ge 0$, we let $\mathcal{M}_{d,j}$ denote the upper triangular matrix  having $1$'s on the diagonal and the only non-zero elements above the diagonal are $(\mathcal{M}_{d,j})_{i,i+1}=q^j$ for all $i$ such that $z_i=d$. We now define the map \[\mathcal{P}_z:A^*\to (\mathbb{N}[q])^{(|z|+1)\times (|z|+1)}, \ w_\ell w_{\ell-1}\cdots w_1 w_0\mapsto \mathcal{M}_{w_\ell,\ell}\cdots \mathcal{M}_{w_0,0}.\]
To avoid cumbersome notation, we do not refer to $z$ when it is clear from the context. For a word~$w$, we say that $\mathcal{P}_z(w)$ is the {\em $q$-Parikh matrix} of $w$ induced by $z$.
\end{definition}

\begin{example}\label{exa:vfirst} With $A=\{1,2,3\}$, $z=12231$, the matrices $\mathcal{M}_{1,j}$, $\mathcal{M}_{2,j}$ and $\mathcal{M}_{3,j}$ are  
\[
{\small\left(
\begin{array}{cccccc}
 1 & q^j & 0 & 0 & 0 & 0 \\
 0 & 1 & 0 & 0 & 0 & 0 \\
 0 & 0 & 1 & 0 & 0 & 0 \\
 0 & 0 & 0 & 1 & 0 & 0 \\
 0 & 0 & 0 & 0 & 1 & q^j \\
 0 & 0 & 0 & 0 & 0 & 1 \\
\end{array}
\right)},\ 
{\small\left(
\begin{array}{cccccc}
 1 & 0 & 0 & 0 & 0 & 0 \\
 0 & 1 & q^j & 0 & 0 & 0 \\
 0 & 0 & 1 & q^j & 0 & 0 \\
 0 & 0 & 0 & 1 & 0 & 0 \\
 0 & 0 & 0 & 0 & 1 & 0 \\
 0 & 0 & 0 & 0 & 0 & 1 \\
\end{array}
\right)},\ 
{\small\left(
\begin{array}{cccccc}
 1 & 0 & 0 & 0 & 0 & 0 \\
 0 & 1 & 0 & 0 & 0 & 0 \\
 0 & 0 & 1 & 0 & 0 & 0 \\
 0 & 0 & 0 & 1 & q^j & 0 \\
 0 & 0 & 0 & 0 & 1 & 0 \\
 0 & 0 & 0 & 0 & 0 & 1 \\
\end{array}
\right)}\]
Letting $w=1212312$, we find $\mathcal{P}_z(w)=\mathcal{M}_{1,6}\mathcal{M}_{2,5}\mathcal{M}_{1,4}\mathcal{M}_{2,3}\mathcal{M}_{3,2}\mathcal{M}_{1,1}\mathcal{M}_{2,0}$ as 
\[
{\small \left(
\begin{array}{cccccc}
 1 & q^6+q^4+q & q^{11}+q^9+q^7+q^6+q^4+q & q^{14}+q^{11}+q^9+q^7 & q^{16} & q^{17} \\
 0 & 1 & q^5+q^3+1 & q^8+q^5+q^3 & q^{10} & q^{11} \\
 0 & 0 & 1 & q^5+q^3+1 & q^7+q^5 & q^8+q^6 \\
 0 & 0 & 0 & 1 & q^2 & q^3 \\
 0 & 0 & 0 & 0 & 1 & q^6+q^4+q \\
 0 & 0 & 0 & 0 & 0 & 1 \\
\end{array}
\right)}
\]
\end{example}
The $q$-deformation version of \cite[Thm.~13]{Serbanuta2004} is given below. For an integer $r$, we let $\mathsf{s}(r)$ denote the sum of the first~$r$ integers, i.e., $r(r+1)/2$.
\begin{theorem}\label{thm:first}
  Let $z$ be a word of length $\ell\ge 1$ whose alphabet is $A$. 
   Let $w\in A^*$. The corresponding $(\ell+1)\times(\ell+1)$ $q$-Parikh matrix is such that  
  \begin{itemize}
  \item $(\mathcal{P}_z(w))_{i,j} = 0$, for all $1\le j < i \le \ell+1$,
  \item $(\mathcal{P}_z(w))_{i,i} = 1$, for all $1\le i \le \ell+1$.
  \item Let $r\in\{1,\ldots,\ell\}$. For all $1\le i \le \ell-r+1 $, $(\mathcal{P}_z(w))_{i,i+r} =q^{\mathsf{s}(r-1)} \qbin{w}{z_i\, z_{i+1} \cdots z_{i+r-1}}$. 
\end{itemize}
\end{theorem}

Before proceeding with the proof, we consider an example to illustrate this statement. The matrix obtained in \cref{exa:vfirst} is indeed
\[
{\small \begin{pmatrix}
  1&\qbin{w}{1}&q\qbin{w}{12}&q^3\qbin{w}{122}&q^6\qbin{w}{1223}&q^{10}\qbin{w}{12231}\\
  0&1&\qbin{w}{2}&q\qbin{w}{22}&q^3\qbin{w}{223}&q^6\qbin{w}{2231}\\
  0&0&1&\qbin{w}{2}&q\qbin{w}{23}&q^3\qbin{w}{231}\\
  0&0&0&1&\qbin{w}{3}&q\qbin{w}{31}\\
  0&0&0&0&1&\qbin{w}{1}\\ 0&0&0&0&0&1\\
\end{pmatrix}}.
\]

We will often make use of the following combinatorial interpretation of the $q$-binomials (and in particular in the proof of \cref{thm:first}), we recall the following statement for the sake of completeness.
\begin{theorem}[{\cite[Thm.~4.1]{RRW}}]\label{thm:powers} Let $u=u_n\cdots u_1$ and $v=v_k\cdots v_1$ be words. Then
  \[
  \qbin{u}{v}=\sum_{Y \in A_{n,k}} q^{{\mathsf{s}}(Y)-\frac{k(k+1)}{2}}
  \]
  where $A_{n,k} = \{ n\ge y_k > \cdots > y_1\ge 1 \mid u_{y_k}\cdots u_{y_1} = v \}$ and for all $Y\in A_{n,k}$, ${\mathsf{s}}(Y):=\sum_{y\in Y} y$.
\end{theorem}
 It is often convenient to express the above result as follows. Any specific occurrence of $v$ as a subword of $u$ contributes to $\qbin{u}{v}$ with a term $q^\alpha$ where $\alpha$ is the sum over all letters of $v$ of the number of letters at the right of them and not being part of that specific occurrence of the subword~$v$. In other words, if
 $Y = \{n \geq y_{k} > \cdots > y_1 \geq 1\} \in A_{n,k}$ corresponds to an occurrence of $v$ in $u$, then $\alpha = \sum_{i=1}^k (y_i - i)$.

 \begin{example}
  Since the argument will be repeated many times, we give an example. Consider the word $u=12\underline{2}1\underline{1}21\underline{2}1=12\underline{2}x\underline{1}y\underline{2}z$ and the subword $v=212$. For the highlighted occurrence of $v$ in $u$, we get a term $q^\alpha$ with $\alpha=|xyz|+|yz|+|z|=|x|+2|y|+3|z|=8$. Indeed, to the right of the leftmost~$2$ of $v$, we see $xyz$ in $u$ (not taking into account the suffix $12$ of $v$). To the right of $1$ in $v$, we see $yz$. Finally, to the right of the rightmost~$2$ of $v$, we see $z$. Another way of obtaining the same counting is to notice that $|x|$ will be counted once because of the prefix $2$ of $v$ to its left. Then $|y|$ will be counted twice because of the prefix $21$ of $v$ to its left, and finally $|z|$ will be counted three times because $v$ is entirely to its left.
 \end{example}

\begin{proof}[Proof of \cref{thm:first}]
   Proceed by induction on the length of $w$. The result trivially holds if $|w|=0,1$. Now assume that the result holds for words of length at most $n$ and consider the word $dw$ of length $n+1$ where $d\in A$ and $|w|=n$. We have that
  \[M=\mathcal{P}_z(dw)=\mathcal{M}_{d,|w|}\mathcal{P}_z(w).\]
  Let $r\in\{1,\ldots,\ell\}$ and $1\le i \le \ell-r+1$. If $z_i\neq d$, then $M_{i,i+r}=[\mathcal{P}_z(w)]_{i,i+r}$ and we may apply the induction hypothesis.

  Now, if $z_i=d$, then
  \[M_{i,i+r}=[\mathcal{P}_z(w)]_{i,i+r}+q^{|w|}[\mathcal{P}_z(w)]_{i+1,i+r}.\]
  By the induction hypothesis, we get
  \[M_{i,i+r}=q^{\mathsf{s}(r-1)} \qbin{w}{z_i\, z_{i+1} \cdots z_{i+r-1}} + q^{|w|}q^{\mathsf{s}(r-2)} \qbin{w}{z_{i+1}\cdots z_{i+r-1}}.\]
  The conclusion follows from the fact that
  \[\qbin{dw}{dz_{i+1}\cdots z_{i+r-1}}=\qbin{w}{dz_{i+1}\cdots z_{i+r-1}}+q^{|w|-(r-1)}\qbin{w}{z_{i+1}\cdots z_{i+r-1}},\]
  which can be deduced using \cref{thm:powers}.
\end{proof}

We recall the following definition.
\begin{definition}
    The \emph{Hadamard product} of two $m\times p$ matrices $A$ and $B$ is the matrix $A\odot B$ whose elements are defined by
    \[
    [A\odot B]_{i,j} = A_{i,j}\cdot B_{i,j}, \quad 1\leq i\leq m, 1\leq j\leq p.
    \]
    Otherwise stated, the Hadamard product of two matrices corresponds to the element-wise product of these matrices.
\end{definition}

\cref{thm:first} permits us to see that $\mathcal{P}_z(w)$ can be expressed as the Hadamard product of two upper triangular matrices, one made of $q$-binomials $\qbin{w}{z_i\cdots z_{i+r-1}}$ and one containing powers of $q$.

\begin{remark}
  In the context of (classical) Parikh matrices, Mateescu's original definition is a special case of \c{S}erb\u{a}nu\c{t}\u{a}'s one. The situation is similar in our context. 
  With the word $z=12\cdots k$ made of all letters of an alphabet, the matrices $\mathcal{M}_{d,j}$ of \cref{def:ourdef} are the ones introduced by E{\u{g}}ecio{\u{g}}lu~\cite{Eeciolu2004}. See, for instance, \cref{tab:ege}.
\end{remark}

\begin{definition}\label{def:Ek}
  In the special case described in the above remark, we get back to E{\u{g}}ecio{\u{g}}lu's considerations. It deserves a particular notation. We let $\mathcal{E}_k:A^*\to(\mathbb{N}[q])^{(k+1)\times (k+1)}$ denote the map $\mathcal{P}_z$ when $z=12\cdots k$.  
\end{definition}

\begin{remark}
The function $\mathcal{P}_z$ is not a morphism from $A^*$ to $\N[q]^{(|z|+1)\times(|z|+1)}$. For instance, $\mathcal{P}_{12}(12).\mathcal{P}_{12}(2)\neq \mathcal{P}_{12}(1).\mathcal{P}_{12}(22)$. But as already observed in \cite{Eeciolu2004}, for all~$z$, $\mathcal{P}_z$ is injective: if $w=w_\ell\cdots w_0\in A^*$ then $\qbin{w}{d}$ has a non-zero term $q^j$ if and only if $w_j=d$. So the second diagonal of $\mathcal{P}_z(w)$ completely determines $w$.
\end{remark}

\subsection{Inverse of $q$-Parikh matrices}
In the case where $z=12\cdots k$, there is a link between the inverse of a Parikh matrix of a word $u$ and the Parikh matrix of its reversal $\widetilde{u}$, see \cite[Thm.~3.2]{MateescuSalomaa2001}. This relationship still holds for (classical) Parikh matrices induced by a word $z=z_1\cdots z_\ell$, where $z_i\neq z_{i+1}$ for all $i\in\{1,\ldots,\ell-1\}$, see \cite[Thm.~16]{Serbanuta2004}. 

For $q$-Parikh matrices, we obtain a similar result. Note that in the next statement, $M_{i,j}(1/q)$ means that the polynomial entry $M_{i,j}$ is evaluated at $1/q$. This is also the reason why we explicitly write a symbol $\cdot$ limiting the scope of this evaluation. As a special case, note that the word $z=12\cdots k$ trivially satisfies the assumption of the theorem.
\begin{theorem}\label{thm:inverse2}
Let $z=z_1\cdots z_\ell$ be a word whose alphabet is $A$ and such that $z_i\neq z_{i+1}$, for all $1\le i<\ell$.
  Let $u\in A^*$. Let $N=(\mathcal{P}_z(u))^{-1}$ be the inverse of the $q$-Parikh matrix of $u$ and $M=\mathcal{P}_z(\widetilde{u})$ be the $q$-Parikh matrix of the reversal of $u$, 
  \begin{itemize}
  \item $N_{i,j} = 0$, for all $1\le j < i \le \ell+1$,
  \item $N_{i,i} = 1$, for all $1\le i \le \ell+1$.
  \item With $1\le i<j \le \ell+1$, $N_{i,j} =(-1)^{i+j} q^{(j-i)(|u|-1)} \cdot M_{i,j}\left(\frac1q\right)$. 
  \end{itemize}
\end{theorem}

\begin{proof}
  For the sake of readability, we assume that $z=12\cdots k$: we get simpler notation but the arguments are the same. Moreover, we will highlight where the assumption on $z$ plays a role.
  
  Proceed by induction on the length of $u$. If $|u|=1$, then $u=d$ for some $d\in\{1,\ldots,k\}$ and $M=\mathcal{P}(\widetilde{u})=\mathcal{M}_{d,0}$, $N=(\mathcal{M}_{d,0})^{-1}$. Except on the diagonal, the only non-zero element of $N$ is $N_{d,d+1}=-1$ and the statement holds true.

  Assume that the result holds for words of length at most $n$ and consider the word $dw$ of length $n+1$ where $d\in\{1,\ldots,k\}$ and $|w|=n$. We have
  \begin{equation}
    \label{eq:productmat}
    (\mathcal{P}(dw))^{-1}=(\mathcal{M}_{d,|w|}\mathcal{P}(w))^{-1}=(\mathcal{P}(w))^{-1}\mathcal{M}_{d,|w|}^{-1}.
  \end{equation}
  Except on the diagonal, the only non-zero element of $\mathcal{M}_{d,|w|}^{-1}$ is \[\left[\mathcal{M}_{d,|w|}^{-1}\right]_{d,d+1}=-q^{|w|}.\]
  For an arbitrary word $z$ satisfying the assumption of the statement, a similar observation holds. For $d\in A$ and $j\ge 0$, the inverse of $\mathcal{M}_{d,j}$ is an upper triangular matrix with the same structure: It has $1$'s on the diagonal and the only non-zero elements above the diagonal are $(\mathcal{M}_{d,j})_{i,i+1}=-q^j$ for all $i$ such that $z_i=d$. This follows from the fact that on the second diagonal two non-zero elements are not consecutive. We discuss the general situation in \cref{rem:no_aa}.
  
  Let $1\le i<j \le k+1$. Assume first $j\neq d+1$, then considering the product \cref{eq:productmat} and using the induction hypothesis, we get
  \[[(\mathcal{P}(dw))^{-1}]_{i,j}=[(\mathcal{P}(w))^{-1}]_{i,j}=(-1)^{i+j} q^{(j-i)(|w|-1)} \cdot [\mathcal{P}(\widetilde{w})]_{i,j}\left(\tfrac1q\right)\]
  On the other hand, by \cref{thm:first}
  \[[\mathcal{P}(\widetilde{w})]_{i,j}=q^{\mathsf{s}(j-i-1)}\qbin{\widetilde{w}}{i\cdots(j-1)}\]
  and from the definition of $q$-binomials
  \begin{equation}
    \label{eq:qbin}
     \qbin{\widetilde{dw}}{i\cdots(j-1)}=\qbin{\widetilde{w}d}{i\cdots(j-1)}=q^{j-i}\qbin{\widetilde{w}}{i\cdots(j-1)}+\underbrace{\delta_{d,j-1}}_{=0}\qbin{\widetilde{w}}{i\cdots (j-2)}.
  \end{equation}
  Consequently,
  \[[\mathcal{P}(\widetilde{dw})]_{i,j}=q^{\mathsf{s}(j-i-1)} q^{j-i}\qbin{\widetilde{w}}{i\cdots(j-1)}=q^{j-i}[\mathcal{P}(\widetilde{w})]_{i,j}.\]
  Hence
  \[[\mathcal{P}(\widetilde{dw})]_{i,j}\left(\tfrac1q\right)=\left(\tfrac1q\right)^{j-i}[\mathcal{P}(\widetilde{w})]_{i,j}\left(\tfrac1q\right) \]
  and 
  \[[\mathcal{P}(\widetilde{w})]_{i,j}\left(\tfrac1q\right)=q^{j-i}\cdot [\mathcal{P}(\widetilde{dw})]_{i,j}\left(\tfrac1q\right)\]
  and the conclusion follows.

  Now consider the case $j=d+1$,  the product \cref{eq:productmat} provides us with two terms and we make again use of the induction hypothesis, $[(\mathcal{P}(dw))^{-1}]_{i,j}$ is equal to 
  \begin{equation}
    \label{eq:again}
    -q^{|w|}(-1)^{i+j-1} q^{(j-1-i)(|w|-1)} \cdot [\mathcal{P}(\widetilde{w})]_{i,j-1}\left(\tfrac1q\right)
   +(-1)^{i+j} q^{(j-i)(|w|-1)} \cdot [\mathcal{P}(\widetilde{w})]_{i,j}\left(\tfrac1q\right).
  \end{equation}
   For convenience, we rewrite this as
    \[[(\mathcal{P}(dw))^{-1}]_{i,j}=(-1)^{i+j} q^{(j-i)(|w|-1)+1} \cdot [\mathcal{P}(\widetilde{w})]_{i,j-1}\left(\tfrac1q\right)
  +(-1)^{i+j} q^{(j-i)(|w|-1)} \cdot [\mathcal{P}(\widetilde{w})]_{i,j}\left(\tfrac1q\right).\]
  
  By \cref{thm:first} and several times using the trick that $q^{-k}\cdot P(1/q)=(q^k P)(1/q)$, we get
  \begin{align*}
    [(\mathcal{P}(dw))^{-1}]_{i,j}&=(-1)^{i+j} q^{(j-i)(|w|-1)} \cdot
                                      \left[q^{\mathsf{s}(j-i-2)}
                                      \left(q^{-1}\qbin{\widetilde{w}}{i\cdots(j-2)}
                                      +q^{j-i-1}\qbin{\widetilde{w}}{i\cdots(j-1)}\right)\right] \left(\tfrac1q\right)\\
    &=(-1)^{i+j} q^{(j-i)(|w|-1)} \cdot
                                      \left[q^{\mathsf{s}(j-i-2)}
                                      q^{-1}\qbin{\widetilde{dw}}{i\cdots(j-1)}
        \right]\left(\tfrac1q\right)\\
                                  &=(-1)^{i+j} q^{(j-i)(|\widetilde{dw}|-1)} \cdot 
                                      \left[q^{\mathsf{s}(j-i-1)}
                                      \qbin{\widetilde{dw}}{i\cdots(j-1)}
 \right]\left(\tfrac1q\right),
  \end{align*}
  where on the second line, we use \cref{eq:qbin}. We conclude by \cref{thm:first}.
\end{proof}

\begin{remark}\label{rem:no_aa}
The reader may wonder why the previous result is not extended to matrices induced by an arbitrary word. The reason is the following one. If $z$ contains a block of consecutive identical letters~$d$, then the inverse of $\mathcal{M}_{d,j}$ which contains a diagonal block of the form (for instance with a factor $ddd$)
\[{\small\left(
\begin{array}{cccc}
 1 & q^j & 0 & 0 \\
 0 & 1 & q^j & 0 \\
 0 & 0 & 1 & q^j \\
 0 & 0 & 0 & 1 \\
\end{array}
\right)}\]
has a diagonal block
\[{\small\left(
\begin{array}{cccc}
 1 & -q^j & q^{2 j} & -q^{3 j} \\
 0 & 1 & -q^j & q^{2 j} \\
 0 & 0 & 1 & -q^j \\
 0 & 0 & 0 & 1 \\
\end{array}
\right)}\]
and therefore, the expression corresponding to \cref{eq:again} has more than two terms. If there are $t$ consecutive identical letters, then a row of the inverse has up to $t+1$ non-zero entries leading to a sum with $t+1$ terms.
\end{remark}

\subsection{Computing inverse using reversal}
For a square matrix $A \in \Z^{n \times n}$, let $A^{\angle}$ denote its \emph{antitranspose} obtained by mirroring the elements along the antidiagonal; $A^{\angle}_{i,j} = A_{n-j,n-i}$. In the literature on Parikh matrices, the \emph{antitranspose} operation on square matrices appears naturally. \cref{thm:reverse} is an adaptation of \cite[Thm.~17]{Serbanuta2004}.

For the sake of completeness, we recall a result from \cite[Cor.~4.6]{RRW} which we will often use. As usual, for a polynomial $P\in\mathbb{C}[q]$ or a formal power series $P\in\mathbb{C}[[q]]$, we let $[q^i]\, P$ denote the coefficient of the term in $q^i$. 
\begin{proposition}\label{cor:reversal}
  Let $u,v$ be words. We have, for all $i\in\{0,\ldots,|v|(|u|-|v|)\}$, 
   \[ [q^i]\qbin{u}{v}=[q^{|v|(|u|-|v|)-i}] \qbin{\widetilde{u}}{\widetilde{v}}.\]
  Otherwise stated,
\[  q^{|v|(|u|-|v|)} \cdot \qbin[1/q]{\widetilde{u}}{\widetilde{v}}
         = \qbin{u}{v}.\]
\end{proposition}

\begin{proposition}\label{prop:reverse}
  Let $z$ be a word of length~$\ell$ and $u$ be a word over the alphabet of $z$. 
 We have for all $i,j\in\{1,\ldots,\ell+1\}$
\[[\mathcal{P}_z(\widetilde{u})]_{i,j} = q^{(j-i)(|u|-1)}\cdot [\mathcal{P}_{\widetilde{z}}(u)]_{\ell+2-j,\ell+2-i}\left(\tfrac1q\right).\]
\end{proposition}
\begin{proof}
  From \cref{cor:reversal}, for $j\ge i+1$, 
  \[\qbin{\widetilde{u}}{z_i\cdots z_{j-1}}=q^{(j-i)(|u|-j+i)}\cdot \qbin[1/q]{u}{z_{j-1}\cdots z_i}.\]
  Hence,
  \[q^{\mathsf{s}(j-i-1)}\qbin{\widetilde{u}}{z_i\cdots z_{j-1}}=q^{(j-i)(|u|-j+i)+2{\mathsf{s}(j-i-1)}}\cdot \left(q^{\mathsf{s}(j-i-1)}\qbin{u}{z_{j-1}\cdots z_i}\right)\left(\tfrac1q\right).\]
  Now observe that
  \[(j-i)(|u|-j+i)+2{\mathsf{s}(j-i-1)}=(j-i)|u|-(j-i)^2+2{\mathsf{s}(j-i-1)}\]
  and $(j-i)^2=2{\mathsf{s}(j-i-1)}+j-i$.
\end{proof}
As an example, with $z=23112311$ and $z=123$, we have
\[\mathcal{P}_z(\widetilde{u})=
{\small\left(
\begin{array}{cccc}
 1 & q^7+q^6+q^3+q^2 & q^{11}+q^{10}+q^7+q^6+q^3+q^2 & q^{12}+q^{11} \\
 0 & 1 & q^4+1 & q^5 \\
 0 & 0 & 1 & q^5+q \\
 0 & 0 & 0 & 1 \\
\end{array}
\right)}\]
and
\[\mathcal{P}_{\widetilde{z}}(u)=
{\small\left(
\begin{array}{cccc}
 1 & q^6+q^2 & q^9 & q^{10}+q^9 \\
 0 & 1 & q^7+q^3 & q^{12}+q^{11}+q^8+q^7+q^4+q^3 \\
 0 & 0 & 1 & q^5+q^4+q+1 \\
 0 & 0 & 0 & 1 \\
\end{array}
\right)}.\]
On the three parallels above the diagonal, the exponent of the convenient power of $q$ is given by $1.(|u|-1)=7$, $2.(|u|-1)=14$ and $3.(|u|-1)=21$.

\begin{theorem}\label{thm:reverse}
  Let $z=z_1\cdots z_\ell$ be a word whose alphabet is $A$ and such that $z_i\neq z_{i+1}$, for all $1\le i<\ell$. Let $u\in A^*$. 
 Let $N=(\mathcal{P}_z(u))^{-1}$ be the inverse of the $q$-Parikh matrix of $u$ and $R=\mathcal{P}_{\widetilde{z}}(u)$ be the $q$-Parikh matrix of $u$ but associated with the reversal of $z$, 
 for all $1\le i,j \le \ell+1$,
    \[N_{i,j} =(-1)^{i+j} R_{\ell+2-j,\ell+2-i}.\] 
    In other words, $\mathcal{P}_z(u)^{-1} = ((-1)^{i+j})_{i,j} \odot \mathcal{P}_{\widetilde{z}}(u)^{\angle}$.
\end{theorem}
\begin{proof}
  This follows directly from \cref{thm:inverse2,prop:reverse}, for $j\ge i+1$, 
  \[N_{i,j} =(-1)^{i+j} q^{(j-i)(|u|-1)} \cdot \left( q^{(j-i)(|u|-1)} \cdot  [\mathcal{P}_{\widetilde{z}}(u)]_{\ell+2-j,\ell+2-i}\left(\tfrac1q\right)\right)\left(\tfrac1q\right).\]
\end{proof}

\subsection{Some consequences and new relations on $q$-binomials}

\cref{thm:inverse2} has some interesting consequences. To help the reader to grasp the idea, we first consider two particular cases with \cref{eq:firstrel1} and \cref{cor:easy} before stating the main result with \cref{pro:general_relation2}.

Let $u$ be a word and $a,b$ be two distinct letters. Let $A_u$ be the alphabet of $u$. Take a word $z$ containing the elements of $A_u\cup\{a,b\}$ exactly once. Therefore, the assumption of  \cref{thm:inverse2} is satisfied: all letters of $z$ are pairwise distinct. We can moreover assume that $z$ has $ab$ as prefix. 
Computing the element $[\mathcal{P}_z(u).\mathcal{P}_z(u)^{-1}]_{1,3}=0$, with $i<\ell$, we get the identity
\begin{equation}
  \label{eq:firstrel1}
  q\left[\qbin{u}{ab}+\qbin{u}{ba}\right]=\qbin{u}{a}\qbin{u}{b}
\end{equation}
because on the first row of $\mathcal{P}_z(u)$ the first three non-zero elements are $1$, $\qbin{u}{a}$, $q\qbin{u}{ab}$ and by \cref{thm:inverse2} the corresponding three non-zero elements on the third column of $\mathcal{P}_z(u)^{-1}$  are
$q^{2(|u|-1)}\cdot \left[q\qbin{\widetilde{u}}{ab}\right]\left(\frac1{q}\right)$, $-q^{|u|-1}\cdot \qbin[1/q]{\widetilde{u}}{b}$, and $1$, but by \cref{cor:reversal},
\[q^{|u|-1}\cdot \qbin[1/q]{\widetilde{u}}{b}=\qbin{u}{b}\]
and
\[q^{2(|u|-2)}\cdot \qbin[1/q]{\widetilde{u}}{ab}=\qbin{u}{ba}.\]
So \cref{eq:firstrel1} holds true and we have obtained this relation without relying on the combinatorial interpretation of the coefficients. The second identity is obtained by considering the element $1,4$ and assuming that $z$ has $abc$ as prefix.
\begin{corollary}\label{cor:easy}
  Let $a,b,c$ be letters such that $a\neq b$ and $b\neq c$. We have
  \[\qbin{u}{a}\qbin{u}{cb}+q^2\qbin{u}{abc}=
  \qbin{u}{c}\qbin{u}{ab}+q^2\qbin{u}{cba}.\]
\end{corollary}
More generally, we have the following result.
\begin{proposition}\label{pro:general_relation2}
    Let $z=z_1\cdots z_n$ be a word over $A$ where $z_i\neq z_{i+1}$ for all $1\le i<n$. For all words~$u$, we have
    \[
    \sum_{\substack{z=xy\\ x,y\in A^*}} (-1)^{|y|}q^{\mathsf{s}(|x|-1)+\mathsf{s}(|y|-1)}\qbin{u}{x}\qbin{u}{\widetilde{y}}=0
  \]
  or equivalently,
  \[\sum_{i=0}^{n} (-1)^{n-i} q^{\mathsf{s}(i-1)+\mathsf{s}(n-i-1)}\qbin{u}{z_1\cdots z_i}\qbin{u}{z_n\cdots z_{i+1}}=0.\]
\end{proposition}

\begin{proof}
If $u$ contains some letters $d_1$, \ldots, $d_s$ not appearing in $z$, instead of $z$ we consider $z'=z d_1\cdots d_s$ (otherwise, we set $z'=z$). Since we append new letters to $z$, this word~$z'$ still satisfies the assumption of \cref{thm:inverse2}.
  
    The result comes directly from the computation of the element $[\mathcal{P}_{z'}(u).\mathcal{P}_{z'}(u)^{-1}]_{1,n+1}=0$. Indeed, carrying out the matrix product and using \cref{thm:first} and \cref{thm:inverse2}, since $[\mathcal{P}_{z'}(u)^{-1}]_{i,n+1}=0$ whenever $i>n+1$, we get
    \begin{align*}
        [\mathcal{P}_{z'}(u).\mathcal{P}_{z'}(u)^{-1}]_{1,n+1} & =\sum_{i=1}^{n+1}[\mathcal{P}_{z'}(u)]_{1,i}\cdot[\mathcal{P}_{z'}(u)^{-1}]_{i,n+1}\\
        & =[\mathcal{P}_{z'}(u)^{-1}]_{1,n+1} + \sum_{i=2}^n[\mathcal{P}_{z'}(u)]_{1,i}\cdot[\mathcal{P}_{z'}(u)^{-1}]_{i,n+1} + [\mathcal{P}_{z'}(u)]_{1,n+1}\\
        & = (-1)^n q^{n(|u|-1)} [\mathcal{P}_{z'}(\widetilde{u})]_{1,n+1}\left(\tfrac1q\right)\\
        & \qquad + \sum_{i=2}^n q^{\mathsf{s}(i-2)} \qbin{u}{z_1 \cdots z_{i-1}}\cdot (-1)^{i+n+1} q^{(n+1-i)(|u|-1)} \cdot [\mathcal{P}_{z'}(\widetilde{u})]_{i,n+1}\left(\tfrac1q\right)\\
        & \qquad \qquad + q^{\mathsf{s}(n)}\qbin{u}{z_1 \cdots z_n}.
    \end{align*}
    Let us focus on one particular term of the sum. For all $i\in\lbrace 1,\ldots,n\rbrace$, we have
    \begin{align*}
        q^{(n+1-i)(|u|-1)} \cdot [\mathcal{P}_{z'}(\widetilde{u})]_{i,n+1}\left(\tfrac1q\right) & = q^{(n+1-i)(|u|-1)} \cdot \left( q^{\mathsf{s}(n-i)} \qbin{\widetilde{u}}{z_i\cdots z_n} \right)\left(\tfrac1q\right)\\
        & = q^{-\mathsf{s}(n-i)}q^{(n+1-i)(|u|-1)} \cdot \qbin[1/q]{\widetilde{u}}{z_i\cdots z_n}\\
        & = q^{-\mathsf{s}(n-i)}q^{(n-i)(n+1-i)}q^{(n+1-i)(|u|-(n+1-i))} \cdot \qbin[1/q]{\widetilde{u}}{z_i\cdots z_n}\\
        & = q^{-\mathsf{s}(n-i)}q^{(n-i)(n+1-i)}\qbin{u}{z_n\, z_{n-1}\cdots z_i}\\
        & = q^{\mathsf{s}(n-i)}\qbin{u}{z_n\, z_{n-1}\cdots z_i},
    \end{align*}
    where the second to last equality comes from \cref{cor:reversal}. Hence,
    \begin{align*}
        [\mathcal{P}_{z'}(u).\mathcal{P}_{z'}(u)^{-1}]_{1,n+1} & = (-1)^n q^{\mathsf{s}(n-1)}\qbin{u}{z_n\, z_{n-1}\cdots z_1}\\
        & \qquad + \sum_{i=2}^n (-1)^{i+n+1} q^{\mathsf{s}(i-2)} \qbin{u}{z_1 \cdots z_{i-1}} q^{\mathsf{s}(n-i)}\qbin{u}{z_n\, z_{n-1}\cdots z_i}\\
        & \qquad \qquad + q^{\mathsf{s}(n)}\qbin{u}{z_1 \cdots z_n}.
    \end{align*}
    For a given factorization $z=xy$, let us denote $x=z_1\cdots z_{i-1}$ and $y=z_i\cdots z_n$, and notice that
    \[
    n+1+i=|z|+1+|x|-1=|z|+|x|=2|z|-|y|,
    \]
    so that $(-1)^{i+n+1}=(-1)^{|y|}$. We can finally rewrite the above sum as
    \[
    [\mathcal{P}_{z'}(u).\mathcal{P}_{z'}(u)^{-1}]_{1,n+1} = \sum_{\substack{z=xy\\ x,y\in A^*}} (-1)^{|y|}q^{\mathsf{s}(|x|-1)+\mathsf{s}(|y|-1)}\qbin{u}{x}\qbin{u}{\widetilde{y}}.
    \]
\end{proof}

\section{Convergence to a formal power series}\label{sec:convergence}

We first observe that with any left-infinite word $\mathbf{x}=\cdots x_2x_1x_0$, if $p_n=x_{n-1}\cdots x_0$ is the prefix of length $n$ of this word, then the polynomial sequence $n\mapsto\qbin{p_n}{z}$ converges to a formal power series $\mathfrak{s}_{\mathbf{x},z}$. If the infinite word is moreover $k$-automatic, then by classical arguments the limit series is shown to be $k$-regular.

To make a connection with what is known in the classical setting, Salomaa has shown that the integer sequence $n\mapsto\binom{u^n}{z}$ satisfies a linear recurrence relation \cite{Salomaa2008}. It is therefore natural to ask what more can be said about the sequence of polynomials $n\mapsto\qbin{u^n}{z}$. As a generalization of Salomaa's result, we show that it satisfies a linear recurrence with polynomial coefficients. Finally for a periodic infinite word $\mathbf{u}=\cdots uuu$, we have a precise description of the growth order of the coefficients of the series $\mathfrak{s}_{\mathbf{u},z}$.

\subsection{Convergence and automaticity}

Due to \cref{def:qbin} and relation \cref{eq:recdef} where the focus is on the rightmost letter, it is more convenient to consider a left-infinite word $\cdots x_n x_{n-1}\cdots x_0=\mathbf{x}$. We say that $x_{t-1}\cdots x_0$ is the {\em prefix} of length~$t$ of $\mathbf{x}$. 

\begin{proposition}\label{pro:defser}
  Let $\cdots x_n x_{n-1}\cdots x_0=\mathbf{x}$ be a left-infinite word and $z$ be a finite word. For all $r\ge 0$, there exists $N$ such that the coefficient of $q^r$ in the $q$-binomials $\qbin{ x_n \cdots x_0}{z}$ is the same for all $n\ge N$. Otherwise stated, the sequence $\left([q^r]\qbin{ x_n \cdots x_0}{z}\right)_{n\ge 0}$ is eventually constant.
\end{proposition}

\begin{proof}
  If $z$ occurs as a subword of some prefix of length~$N$ of $\mathbf{x}$, then it will occur within all longer prefixes of $\mathbf{x}$. Let $n\ge N-1$. Take a specific occurrence of $z=z_1\cdots z_\ell$ such that
  \[x_n\cdots x_0=z_1y_1z_2y_2\cdots z_\ell y_\ell,\]
where $y_1$, \ldots, $y_\ell\in A^*$. By \cref{thm:powers} it provides $\qbin{x_n\cdots x_0}{z}$ with some monomial $q^r$ where
  \[r=\sum_{i=1}^\ell i\, |y_i|\ge |y_1\cdots y_\ell|=n+1-|z|.\]
 Roughly speaking, the further to the left occurs the first letter of $z$ (i.e., the larger $n$ is such that $x_n=z_1$), the larger the corresponding exponent is.

  Instead of focusing on a specific occurrence of $z$, let us focus on a specific exponent of the $q$-binomial. From the above discussion, for each given~$r$ and for large enough~$n$, the coefficient of~$q^r$ in $\qbin{x_n\cdots x_0}{z}$ is completely determined by a suffix of $\mathbf{x}$ whose length is bounded by~$r+|z|$. Hence, for all $n\ge r+|z|-1$, the polynomials $\qbin{x_n\cdots x_0}{z}$ have the same coefficient for $q^r$.
\end{proof}

This result legitimates the next definition.

\begin{definition}
 Let $\mathbf{x}$ be a left-infinite word and $z$ be a finite word. For all $r\ge 0$, we let $c_r$ denote the coefficient of $q^r$ in the $q$-binomials $\qbin{ x_n \cdots x_0}{z}$ for all $n\ge r+|z|-1$. We let $\mathfrak{s}_{\mathbf{x},z}$ denote the formal series defined by $\sum_{r\ge 0}c_r\, q^r$.  
\end{definition}


As an example, consider the (left) Thue--Morse sequence $\cdots110010110=\mathbf{t}$ and the word $z=00$. We obtain the (right-infinite) sequence of coefficients whose first terms are 
\[00101101211211412313324323525505635534844655764765957847\cdots\]
i.e., 
\[\mathfrak{s}_{\mathbf{t},00}=q^2+q^4+q^5+q^7+2q^8+q^9+\cdots\]
For instance, this means that, for large enough~$n$, in every $q$-binomial $\qbin{t_n\cdots t_0}{00}$ the monomial of least degree with a non-zero coefficient is $q^2$. 
This is (up to a shift) the sequence {\tt A133009} from OEIS
counting the number $c(n)$ of pairs $(x,y)$ of integers such as their base-$2$ expansions contain an odd number of ones (the Thue--Morse word is the characteristic sequence of the set of integers with this property), $x<y$ and $x+y=n$. For instance, $c(3)=1$ because there is only one pair $(1,2)$ such that $1+2=3$. This sequence was introduced in \cite{ErdosIV} and the connection with automatic sequences was studied in \cite{allouche2022additive}.

We assume the reader familiar with $k$-automatic and $k$-regular sequences  \cite{AS,shallit2022logical}. For instance, the Thue--Morse sequence is $2$-automatic.  When the infinite word considered in \cref{pro:defser} is $k$-automatic, we obtain a $k$-regular sequence of coefficients.

\begin{proposition}
  Let $\cdots x_2x_1x_0=\mathbf{x}$ be a (left-infinite) $k$-automatic sequence over an alphabet $A$ and $z$ be a word. The sequence of coefficients of the series $\mathfrak{s}_{\mathbf{x},z}$ is $k$-regular.
\end{proposition}

\begin{proof}
  Since $\mathbf{x}$ is $k$-automatic, for all $a\in A$, there is a first-order formula $\varphi_a(i)$ in $\langle\mathbb{N},+,V_k\rangle$ which holds true if and only if $x_i=a$. Let $z=z_\ell\cdots z_1$. The next formula permits to detect occurrences of $z$ as a subword of $\mathbf{x}$:
  \[\Psi_z(i_\ell,\ldots,i_1) \equiv i_\ell>\cdots>i_1 \wedge \varphi_{z_\ell}(i_\ell) \wedge \cdots \wedge \varphi_{z_1}(i_1)\]
Since $z$ is given, note that $\mathsf{s}(\ell-1)$ is a constant. For all $n\ge 0$, the set
\[T_n:=\{(i_\ell,\ldots,i_1) \mid \Psi_z(i_\ell,\ldots,i_1) \wedge n=i_\ell+\cdots+i_1-\mathsf{s}(\ell-1)\}\]
is definable in $\langle\mathbb{N},+,V_k\rangle$. By \cref{thm:powers} the number of $\ell$-tuples in $T_n$ gives the coefficient of $q^n$ in the series
\[\# T_n= [q^n] \mathfrak{s}_{\mathbf{x},z}.\]
By classical enumeration arguments about $k$-automatic formulas (for a proof, see \cite{CRS} ; for details, see \cite[Chap.~9]{shallit2022logical} or \cite[Sec.~2]{allouche2022additive}), the sequence $n\mapsto \# T_n$ is $k$-regular.
\end{proof}

\begin{corollary}[{\cite[Thm.~9.7.1]{shallit2022logical}}]\label{cor:reg-growth}
   Let $\mathbf{x}$ be a (left-infinite) $k$-automatic sequence and $z$ be a word. There exists a real number $\alpha>0$ such that $[q^n] \mathfrak{s}_{\mathbf{x},z} $ is in $\mathcal{O}(n^\alpha)$.
\end{corollary}


\subsection{A fine analysis of the periodic case}

Let $u,z$ be finite non-empty words. The left-infinite periodic word $\mathbf{u}=\cdots uuu$ is $k$-automatic for all $k\ge 2$, hence the series $\mathfrak{s}_{\mathbf{u},z}$ is $k$-regular for all $k\ge 2$. A generalization of Cobham's theorem implies that the sequence of coefficients of $\mathfrak{s}_{\mathbf{u},z}$ satisfies a linear recurrence relation \cite{Bell2005}.

Our aim in this section is to get a precise description of the polynomial $\qbin{u^n}{z}$ not reducing ourselves to the limit case given by $\mathfrak{s}_{\mathbf{u},z}$. From this, we obtain a generalization of Salomaa's result with \cref{cor:gen_sal}: The sequence of $q$-binomials $(\qbin{u^n}{z})_{n\ge 0}$ satisfies a linear recurrence relation over $\mathbb{Z}[q]$. In particular, our developments are here independent of the general theory of automatic and regular sequences.

\begin{definition}
Given two words $u,z$ where $u$ is over the alphabet $A$ of $z$ and a non-negative integer $k$, let us define the square matrix Pow$_z(k)$ of dimension $|z|+1$ as follows:
\[
[\Pow_z(k)]_{i,j}=
\begin{cases}
0 & \text{if } i>j\\
1 & \text{if } i=j\\
q^{(j-i)k} & \text{if } i<j,\\
\end{cases}
\quad \text{ i.e., }\quad 
\Pow_z(k) =
{\small\begin{pmatrix}
    1      & q^k     & q^{2k} & \cdots & q^{|z|k}\\
    0      &    1    & q^k    & \ddots & \vdots\\
    \vdots &  \ddots & \ddots & \ddots & q^{2k}\\
    \vdots &         & \ddots &    1   & q^k\\
    0      &  \cdots & \cdots &    0   & 1\\
\end{pmatrix}}.
\]
Finally, we let $H_{z,k}(u)$ denote the Hadamard product $\mathcal{P}_z(u)\odot\Pow_z(k)$.
\end{definition}

The following result allows us to express the Parikh matrix of $u^n$ induced by $z$ using the matrices $H_{z,k}(u)$ we just defined.
\begin{lemma}\label{lemma:Hmatrix}
    Let $k\in\N$, and let us write $u=u_\ell\cdots u_0$ with $u_i\in A$ for all $i\in\lbrace 0,\ldots,\ell\rbrace$. We have
    \[
    H_{z,k}(u) = \mathcal{M}_{u_\ell,k+\ell}\cdot\mathcal{M}_{u_{\ell-1},k+\ell-1}\cdots\mathcal{M}_{u_1,k+1}\cdot\mathcal{M}_{u_0,k}.
    \]
\end{lemma}
\begin{proof}
    We proceed by induction on $|u|\geq 1$. If $|u|=1$, then $u=a\in A$ and, using the definition of the Parikh matrix $\mathcal{P}_z(a)$, we have
    \begin{align*}
    H_{z,k}(a) & = \mathcal{P}_z(a)\odot\Pow_z(k) = \mathcal{M}_{a,0}\odot\Pow_z(k) = \mathcal{M}_{a,k}.
    \end{align*}
    
    Let us now assume that the result is true for all words $u$ of length $\ell\geq 1$, and show that it still holds for a word $w=w_\ell\cdots w_0$ of length $\ell+1$. If $|w|=\ell+1$, then there exist $u\in A^*$ and $a\in A$ such that $|u|=\ell$ and $w=au=au_{\ell-1}\cdots u_0$. By induction hypothesis, 
    \begin{align*}
    \mathcal{M}_{w_\ell,k+\ell}\cdot\mathcal{M}_{w_{\ell-1},k+\ell-1}\cdots\mathcal{M}_{w_0,k} & = \mathcal{M}_{a,k+\ell}\cdot\mathcal{M}_{u_{\ell-1},k+\ell-1}\cdots\mathcal{M}_{u_0,k}\\
    & = \mathcal{M}_{a,k+\ell}\cdot H_{z,k}(u).
    \end{align*}
    We thus have to show that $\mathcal{M}_{a,k+\ell}\cdot H_{z,k}(u)=H_{z,k}(w)$. By definition,
    \[
    [\mathcal{M}_{a,k+\ell}]_{i,j}=
    \begin{cases}
    1,       & \text{if } i=j;\\
    q^{k+\ell}, & \text{if } i+1=j \text{ and } z_i=a;\\
    0,       & \text{otherwise};
    \end{cases}
    \]
    and
    \[
    [H_{z,k}(u)]_{i,j}=
    \begin{cases}
    0, & \text{if } i>j;\\
    1, & \text{if } i=j;\\
    q^{(j-i)k+\mathsf{s}(j-i-1)}\qbin{u}{z_i\cdots z_{j-1}}, & \text{if } i<j.
    \end{cases}
    \]
    It is clear that $[\mathcal{M}_{a,k+\ell}\cdot H_{z,k}(u)]_{i,j}=\delta_{i,j}$ for every pair $(i,j)$ such that $1\leq j\leq i\leq |z|+1$. Now assume that $i<j$, and let us consider two cases.
    If $z_i\neq a$, then
    \begin{align*}
    [\mathcal{M}_{a,k+\ell}\cdot H_{z,k}(u)]_{i,j} & = q^{(j-i)k+\mathsf{s}(j-i-1)}\qbin{u}{z_i\cdots z_{j-1}}\\
    & = q^{(j-i)k+\mathsf{s}(j-i-1)}\qbin{au}{z_i\cdots z_{j-1}}\\
    & = [H_{z,k}(w)]_{i,j},
    \end{align*}
    and if $z_i=a$, then
    \begin{align*}
    [\mathcal{M}_{a,k+\ell}\cdot H_{z,k}(u)]_{i,j} & = q^{(j-i)k+\mathsf{s}(j-i-1)}\qbin{u}{z_i\cdots z_{j-1}} + q^{k+\ell}q^{(j-i-1)k+\mathsf{s}(j-i-2)}\qbin{u}{z_{i+1}\cdots z_{j-1}}\\
    & = q^{(j-i)k+\mathsf{s}(j-i-1)}\left( \qbin{u}{z_i\cdots z_{j-1}} + q^{\ell-(j-i-1)}\qbin{u}{z_{i+1}\cdots z_{j-1}} \right)\\
    & = q^{(j-i)k+\mathsf{s}(j-i-1)}\qbin{w}{z_i\cdots z_{j-1}}\\
    & = [H_{z,k}(w)]_{i,j},
    \end{align*}
    where the second to last equality comes from \cref{thm:powers}, since $|u|=\ell$ and $|z_{i+1}\cdots z_{j-1}|=j-i-1$. In any case, we find that $[\mathcal{M}_{a,k+\ell}\cdot H_{z,k}(u)]_{i,j}=[H_{z,k}(w)]_{i,j}$ for all $1\leq i,j\leq |z|+1$, which completes the proof.
\end{proof}

\begin{corollary}\label{cor:matH}
    For any positive integer $n$ and $u\in A^*$, we have
    \[
    \mathcal{P}_z(u^n)=H_{z,(n-1)|u|}(u)\cdot H_{z,(n-2)|u|}(u)\cdots H_{z,|u|}(u)\cdot H_{z,0}(u).
    \]
\end{corollary}
\begin{proof}
    This is an immediate consequence of the definition of $\mathcal{P}_z(u^n)$ and \cref{lemma:Hmatrix}. Let $u=u_\ell\cdots u_0$. Indeed, by definition we have
    \begin{align*}
        \mathcal{P}_z(u^n) & = (\mathcal{M}_{u_\ell,n|u|-1}\cdots\mathcal{M}_{u_0,(n-1)|u|})\cdots(\mathcal{M}_{u_\ell,2|u|-1}\cdots\mathcal{M}_{u_0,|u|})(\mathcal{M}_{u_\ell,|u|-1}\cdots\mathcal{M}_{u_0,0})\\
        & = H_{z,(n-1)|u|}(u)\cdots H_{z,|u|}(u)\cdot H_{z,0}(u).
    \end{align*}
\end{proof}

\begin{theorem}\label{thm:unz}
  The $q$-binomial $\qbin{u^n}{z}$ can be expressed as
  \[\frac{1}{q^{\mathsf{s}(|z|-1)}}\sum_{k=1}^{m} R_{k}(q)\, \frac{1-q^{c_{k}n|u|}}{1-q^{c_{k}|u|}}\] where $m$ and $c_{k}$ are positive integers and $R_{k}$ are rational functions whose denominators only have factors of the form $(1-q^{t|u|})$ for some integer $t$. Moreover, these quantities $c_k$ and $R_k$ can be effectively computed. In particular, the sequence $(\qbin{u^n}{z})_{n\ge 0}$ converges in $\mathbb{N}[[q]]$ to the formal power series $\mathfrak{s}_{\mathbf{u},z}(q)$ expressed by the rational function \[\frac{1}{q^{\mathsf{s}(|z|-1)}}\sum_{k=1}^{m} R_{k}(q)\, \frac{1}{1-q^{c_{k}|u|}}.\]
\end{theorem}
\begin{proof}
For all $n\ge 1$, we let $\overrightarrow{p_n}$ denote the last column of $\mathcal{P}_z(u^n)$, i.e.,
\[\overrightarrow{p_n}=
\begin{pmatrix}
  p_{n,\ell}\\
  p_{n,\ell-1}\\
  \vdots\\
  p_{n,1}\\
  p_{n,0}\\
\end{pmatrix}:=
\begin{pmatrix}
  q^{\mathsf{s}(\ell-1)}\qbin{u^n}{z_1\cdots z_\ell}\\
  q^{\mathsf{s}(\ell-2)}\qbin{u^n}{z_2\cdots z_\ell}\\
  \vdots\\
  \qbin{u^n}{z_\ell}\\
  1\\
\end{pmatrix}.\]
In particular, $\overrightarrow{p_1}$ is the last column of $\mathcal{P}_z(u)$ and, for convenience, we set $\overrightarrow{p_0}$ to be the column vector made of zeroes with only a $1$ in last position. \cref{cor:matH} shows that, for all $n\ge 0$, 
\begin{equation}\label{eq:matH}
  \overrightarrow{p_{n+1}}=H_{z,n|u|}\cdot \overrightarrow{p_n}.
\end{equation}
Since $H_{z,n|u|}$ is upper triangular, we determine the elements of $\overrightarrow{p_n}$ from bottom to top. First,  we have $p_{n,0}=1$ for all $n$. Next, from \cref{eq:matH}
\[p_{n+1,1}=p_{n,1}+\qbin{u}{z_\ell} q^{n|u|}\, p_{n,0}\]
and we deduce from the above recurrence that
\[p_{n+1,1}=\qbin{u}{z_\ell} \sum_{i=0}^{n} q^{i|u|}=\qbin{u}{z_\ell} \frac{1-q^{(n+1)|u|}}{1-q^{|u|}}.\]
Let us produce one extra element, again from \cref{eq:matH}
\begin{eqnarray*}
  p_{n+1,2}&=&p_{n,2}+\qbin{u}{z_{\ell-1}} q^{n|u|}\, p_{n,1}+q\qbin{u}{z_{\ell-1}z_\ell} q^{2n|u|}\, p_{n,0}  \\
           &=& p_{n,2}+\qbin{u}{z_{\ell-1}} \qbin{u}{z_\ell} \frac{1}{1-q^{|u|}} \left( q^{n|u|}-q^{2n|u|}\right)
               +\qbin{u}{z_{\ell-1}z_\ell} q^{2n|u|+1}.
\end{eqnarray*}
We deduce that
\[p_{n+1,2}=\qbin{u}{z_{\ell-1}} \qbin{u}{z_\ell} \left( \frac{1-q^{(n+1)|u|}}{(1-q^{|u|})^2}
  -\frac{1-q^{2(n+1)|u|}}{(1-q^{|u|})(1-q^{2|u|})}\right)+ q\qbin{u}{z_{\ell-1}z_\ell}\frac{1-q^{2(n+1)|u|}}{(1-q^{2|u|})}.\]
To get a better grasp of the statement, let us rewrite the polynomial as 
\[p_{n+1,2}=R_{2,1}(q)\, \frac{1-q^{(n+1)|u|}}{1-q^{|u|}}+R_{2,2}(q)\, \frac{1-q^{2(n+1)|u|}}{1-q^{2|u|}}\]
with the rational functions
\[R_{2,1}(q)=\frac{\qbin{u}{z_{\ell-1}} \qbin{u}{z_\ell} }{1-q^{|u|}},\quad 
R_{2,2}(q)=\frac{-\qbin{u}{z_{\ell-1}} \qbin{u}{z_\ell} + q\qbin{u}{z_{\ell-1}z_\ell}(1-q^{|u|})}{1-q^{|u|}}.\]

Let $1\le j\le\ell$, we prove by induction on $j$ (already shown for $j=1,2$) that 
\[p_{n+1,j}=\sum_{k=1}^{m_j} R_{j,k}(q)\, \frac{1-q^{c_{j,k}(n+1)|u|}}{1-q^{c_{j,k}|u|}}\] where $m_j$ and $c_{j,k}$ are positive integers and $R_{j,k}$ are rational functions whose denominators only have factors of the form $(1-q^{t|u|})$ for some integer $t$.
We have using \cref{eq:matH}
\[p_{n+1,j}=p_{n,j}+\sum_{k=1}^j \qbin{u}{z_{\ell-j+1}\cdots z_{\ell-j+k}} q^{kn|u|+\mathsf{s}(k-1)}p_{n,j-k}\]
by the induction hypothesis, we have already expressed $p_{n,0},\ldots,p_{n,j-1}$ so 
\[p_{n+1,j}=p_{n,j}+\sum_{k=1}^j \qbin{u}{z_{\ell-j+1}\cdots z_{\ell-j+k}} q^{kn|u|+\mathsf{s}(k-1)}\sum_{i=1}^{m_{j-k}} R_{j-k,i}(q)\, \frac{1-q^{c_{j-k,i}n|u|}}{1-q^{c_{j-k,i}|u|}}.\]
Hence $p_{n+1,j}$ is equal to
\begin{equation}
  \label{eq:pnj}
\sum_{k=1}^j q^{\mathsf{s}(k-1)}\qbin{u}{z_{\ell-j+1}\cdots z_{\ell-j+k}} \sum_{i=1}^{m_{j-k}} \frac{R_{j-k,i}(q)}{{1-q^{c_{j-k,i}|u|}}}\,
\left(\frac{1-q^{k(n+1)|u|}}{1-q^{k|u|}}-\frac{1-q^{(c_{j-k,i}+k)(n+1)|u|}}{1-q^{(c_{j-k,i}+k)|u|}}\right)  
\end{equation}
which has the desired form as we now explain.
\begin{eqnarray*}
  p_{n+1,j}&=& \sum_{k=1}^j \overbrace{\sum_{i=1}^{m_{j-k}} q^{\mathsf{s}(k-1)}\qbin{u}{z_{\ell-j+1}\cdots z_{\ell-j+k}}  \frac{R_{j-k,i}(q)}{{1-q^{c_{j-k,i}|u|}}}}\,
                \frac{1-q^{k(n+1)|u|}}{1-q^{k|u|}}\\
               &&+\sum_{k=1}^j  \sum_{i=1}^{m_{j-k}} \underbrace{- q^{\mathsf{s}(k-1)}\qbin{u}{z_{\ell-j+1}\cdots z_{\ell-j+k}} \frac{R_{j-k,i}(q)}{{1-q^{c_{j-k,i}|u|}}}}\,
\frac{1-q^{(c_{j-k,i}+k)(n+1)|u|}}{1-q^{(c_{j-k,i}+k)|u|}}
\end{eqnarray*}
where the braced factors are $j+m_0+\cdots+m_{j-1}$ new rational functions $R_{j,\cdot}$ having, by induction hypothesis, the right form.

To conclude with the proof, $\qbin{u^n}{z}$ is equal up to a multiplicative factor $q^{\mathsf{s}(|z|-1)}$ to the upper right corner of $\mathcal{P}_z(u^n)$ which is equal to $p_{n,\ell}$.
\end{proof}

\begin{example}\label{exa:powunz}
  Let $u=0110$ and $z=01$. We have
  \[H_{z,0}=\mathcal{P}_z(u)=
  \begin{pmatrix}
    1&q^3+1&q^5+q^4\\
    0&1&q^2+q\\
    0&0&1\\
  \end{pmatrix} \text{ and }
  \overrightarrow{p_1}= \begin{pmatrix}
    q^5+q^4\\
    q^2+q\\
    1\\
  \end{pmatrix}.\]
  Hence, with the notation of the previous proof, we can express $\qbin{u^{n+1}}{1}$ as 
  \[p_{n+1,1}=\qbin{z}{1} \frac{1-q^{4(n+1)}}{1-q^{4}} = (q^2+q) \frac{1-q^{4(n+1)}}{1-q^{4}}.\]
  and
  \[p_{n+1,2}=p_{n,2}+\qbin{z}{0}(q^2+q)  q^{4n}\frac{1-q^{4n}}{1-q^4}+q\qbin{z}{01}q^{8n}.\]
  So, we have an expression for $q\qbin{u^{n+1}}{z}$ as 
  \[p_{n+1,2}=(q^3+1)(q^2+q)  \left(\frac{1-q^{4(n+1)}}{(1-q^4)^2}-\frac{1-q^{8(n+1)}}{(1-q^4)(1-q^8)}\right)
  +(q^5+q^4)\frac{1-q^{8(n+1)}}{1-q^8}\]
  which can be rewritten as
  \begin{equation}
    \label{eq:expow}
    p_{n+1,2}=\frac{(q^3+1)(q^2+q)}{1-q^4} \frac{1-q^{4(n+1)}}{1-q^4}
  +\left(\frac{(q^3+1)(q^2+q)}{1-q^4}+(q^5+q^4)\right) \frac{1-q^{8(n+1)}}{1-q^8}
  \end{equation}
  Now consider the corresponding series (obtained by discarding the term $q^{4(n+1)}$ and $q^{8(n+1)}$ appearing on the numerators)
  \[  \frac{(q^3+1)(q^2+q)(q^4-q^8)}{(1-q^4)^2(1-q^8)}
  +\frac{q^5+q^4}{1-q^8}=\frac{q^4}{(q-1)^2
   \left(q^2+1\right)^2
   \left(q^4+1\right)}:=\frac{q^4}{D(q)}.\]
 If we first divide by $q$, the series expansion $\sum_{r\ge 0} c_rq^r$ is of the form
 \[q^3+2 q^4+q^5+q^7+2 q^8+q^9+2
   q^{11}+4 q^{12}+2 q^{13}+2
   q^{15}+4 q^{16}+2 q^{17}+3
   q^{19}+6 q^{20}+3 q^{21}+\cdots\]
   and those coefficients match exactly those of $\qbin{u^{n+1}}{z}$ for large enough $n$. By taking the reciprocal (or reflected polynomial as called in \cite{GKP}) $q^{10}D(1/q)$ of the denominator \[D(q)=q^{10}-2 q^9+3 q^8-4 q^7+4 q^6-4
   q^5+4 q^4-4 q^3+3 q^2-2 q+1,\]
 it is a routine technique to see that $c_n$ satisfies the order-$10$ recurrence relation
 \[c_n=2c_{n-1}-3c_{n-2}+4c_{n-3}-4c_{n-4}+4c_{n-5}-4c_{n-6}+4c_{n-7}-3c_{n-8}+2c_{n-9}-c_{n-10}.\]
 The fact that the coefficient of $q^{2+4n}$ is vanishing will be explained by \cref{rem:vanish}
\end{example}

Since the limit formal power series $\mathfrak{s}_{\mathbf{u},z}$ is a rational function, as we have shown in the above example, it is not surprising that its coefficients satisfy a linear recurrence relation with constant coefficients. We now turn to the polynomial sequence $(\qbin{u^n}{z})_n$ and show that it too satisfies a recurrence relation, but this time with (constant) polynomial coefficients. We make use of classical arguments about linear recurrences. See, for instance \cite{GKP}, for a general reference.

\begin{lemma}\label{lem:rec}
  Let $R_j(q)$ be given rational functions and $c_j$ be pairwise distinct non-negative integers. The sequence $(p_n)_{n\ge 0}$ 
  \[\left( \sum_{j=1}^s R_j(q)\, q^{c_j n}\right)_{n\ge 0}\]
  satisfies the following linear recurrence relation of order $s$ with polynomial coefficients
  \[p_{n+s}=
  \sum_{k=1}^s  \underbrace{(-1)^{k-1}\left(\sum_{1\le i_1<\cdots <i_k\le s} q^{c_{i_1}+\cdots+c_{i_k}}\right)}_{:=\mathfrak{D}_k} p_{n+s-k} \]
\end{lemma}

\begin{proof}
  The linear recurrent sequences $\left(q^{c_j n}\right)_{n\ge 0}$ and $\left(R_j(q) q^{c_j n}\right)_{n\ge 0}$ have both a characteristic polynomial of the form $X-q^{c_j}$. It is a well-known result that the characteristic polynomial of the sum of linear recurrent sequences (with distinct roots) is the product of the corresponding characteristic polynomials. Hence, the characteristic polynomial of the sequence of interest is \[\prod_{j=1}^s \left( X-q^{c_j n}\right)=X^s+\sum_{k=1}^s \mathfrak{D}_k X^{s-k}.\]
\end{proof}

\begin{corollary}\label{cor:rec}
   Let $R_j(q)$ be given rational functions and $c_j$ be pairwise distinct non-negative integers. The sequence $(p_n)_{n\ge 0}$ 
  \[\left( \sum_{j=1}^s R_j(q)\, \sum_{i=0}^n q^{c_j i}\right)_{n\ge 0}=\left( \sum_{j=1}^s R_j(q)\, \frac{1-q^{c_j (n+1)}}{1-q^{c_j}}\right)_{n\ge 0}\]
  satisfies the following linear recurrence relation of order $s+1$ with polynomial coefficients
 \[p_{n+s+1}=
  \sum_{k=1}^{s+1}(\mathfrak{D}_{k}-\mathfrak{D}_{k-1})\, p_{n+k}\]  
  setting $\mathfrak{D}_0=-1$ and $\mathfrak{D}_{s+1}=0$.
\end{corollary}

\begin{proof}
  This follows from the classical result that if $f(t)$ is the rational function equal to the series $\sum_{n\ge 0} u_n\, t^n$ where $u_n$ satisfies a linear recurrence relation, then $\frac{1}{1-t}.f(t)$ encodes the series of the partial sums and one can get a linear relation from the denominator of the series.
\end{proof}

Letting $q=1$ we recover Salomaa's result \cite[Thm.~3]{Salomaa2008} as a special case of our results.

\begin{corollary}\label{cor:gen_sal}
  The sequence of $q$-binomials $(\qbin{u^n}{z})_{n\ge 0}$ satisfies a linear recurrence relation with polynomial coefficients. In particular, the sequence of binomials $(\binom{u^n}{z})_{n\ge 0}$ satisfies a linear recurrence relation with constant coefficients. 
\end{corollary}

\begin{example}
  Let us continue \cref{exa:powunz}. With the notation of \cref{lem:rec} looking at \cref{eq:expow}, we have $c_1=4$ and $c_2=8$, $\mathfrak{D}_1=q^4+q^8$ and $\mathfrak{D}_2=-q^{12}$. Hence, with \cref{cor:rec}, the sequence $\left(\qbin{(0110)^n}{01}\right)_{n\ge 0}$ satisfies the relation
  \[p_{n+3}=(1+q^4+q^8)p_{n+2}-(q^4+q^8+q^{12})p_{n+1}+q^{12}p_n.\]
  Now the integer sequence $\left(\binom{(0110)^n}{01}\right)_{n\ge 0}$ whose first terms are 
  $0, 2, 8, 18, 32, 50, 72, 98, 128, 162, 200$ satisfies the relation
  \[p_{n+3}=3p_{n+2}-3p_{n+1}+p_n.\]
\end{example}

\subsection{Growth of the coefficients of the series in the periodic case}
Let $\mathbf{u}=\cdots uuu$. Consider again the series  $\mathfrak{s}_{\mathbf{u},z}(q)=\sum_{i\ge 0} c_i\, q^i$ which is the limit of the sequence of $q$-binomials $\qbin{u^n}{z}$ considered in \cref{thm:unz}. We have already observed with \cref{cor:reg-growth} that $i\mapsto c_i$ has a polynomial growth. In this section, we obtain some more information about the growth order and the indices for which $c_i$ vanishes.

\begin{lemma}
  Let  $\mathfrak{s}_{\mathbf{u},z}(q)=\sum_{i\ge 0} c_i\, q^i$. The function $i\mapsto c_i$ is in $\mathcal{O}(i^{|z|-1})$.
\end{lemma}

\begin{proof}
  Proceeding as in the proof of \cref{thm:unz}, we can express
  \[\qbin{u^n}{z}=\sum_{j=1}^r \frac{P_j(q)}{\prod_{i=1}^{|z|} (1-q^{j_i|u|})^{\alpha_{j_i}}}\]
  where $P_j$'s are polynomials and for all $i$, $\sum_{i=1}^{|z|} \alpha_{j_i}\le |z|$. Indeed, in the expression \cref{eq:pnj} to get $p_{n,j}$ from $p_{n,j-1}$ in the induction step, a multiplication by $\frac{1}{1-q^{(c_{j-k,i}+k)|u|}}$ occur. So, at each stage, we create a factor  $1/(1-q^{j_i|u|})$, and some of these factors may be equal and collected together thanks to the exponent $\alpha_{j,i}$.
  One concludes by considering the partial fraction decomposition and recalling that the $n$th coefficient in the series expansion of $\frac{1}{(1-q)^t}$ is in $\mathcal{O}(n^{t-1})$.
\end{proof}

\begin{theorem}
  Let  $\mathfrak{s}_{\mathbf{u},z}(q)=\sum_{i\ge 0} c_i\, q^i$. Let $r\in\{0,\ldots,|u|-1\}$. Either $c_{r+i|u|}$ is zero for all large enough $i$, or the growth order of the function $r+i|u| \mapsto c_{r+i|u|}$ is in $\Theta(i^{|z|-1})$. Moreover, the two kinds of behavior are completely determined by the words $u$ and $z$.
\end{theorem}

\begin{proof}[Proof sketch]
The previous lemma already gives the upper bound, so it suffices to give the lower bound of correct order.
Let $z=z_{\ell}\cdots z_1$ with $\ell=|z|$. For $j\in\{1,\ldots,\ell\}$, we let
\begin{equation}
  \label{eq:ups-bis}
  u=p_j z_j s_j \quad \text{and}\quad |s_j|=t_j.
\end{equation}
(Here we assume that each letter of $z$ appears in $u$, as otherwise $\qbin{u^k}{z} = 0$ for all $k \geq 0$.)
These $\ell$ factorizations of $u$ are not necessarily unique. Let us consider one such $\ell$-tuple $(t_\ell,\ldots,t_1)$ of non-negative integers. We will discuss later on the possible choices: each such tuple will provide a periodic sequence of indices of period $|u|$ for which the corresponding coefficients growth polynomially.

\begin{claim}
For all large enough $n$, if $[q^n]\qbin{u^{n+\ell}}{z} \neq 0$, then there exists an occurrence of $z$ in $u^{n+\ell}$
such that the letters of $z$ appear in distinct copies of $u$.
\end{claim}
\begin{claimproof}
Let $t$ be the maximal over all possible $t_j$s as in \eqref{eq:ups-bis}, and assume that $n \geq \ell t + (|u|-1)\frac{\ell(\ell-1)}{2}$.
Let us write $u^{n+\ell} = u_{n+\ell} \cdots u_1$. Since the coefficient is assumed positive, there is an occurrence of $z$ in $u^{n+\ell}$
contributing the monomial $q^n$. Let $m_i$ indicate the copy $u_{m_i}$ in which the letter $z_i$ of this particular occurrence of $z$ appears in,
and let $(t_\ell,\ldots,t_1)$ be the corresponding set of factorizations as in \eqref{eq:ups-bis}.
Then we have that
\begin{equation}\label{eq:exponent-form}
n = \sum_{i=1}^\ell t_i + (|u|-1)\frac{\ell(\ell-1)}{2}  + |u|\cdot \sum_{i=1}^{\ell}(m_i-i).
\end{equation}
Indeed, for each $i=1$, \ldots, $\ell$, the number of letters to the right of $z_i$ no appearing in this occurrence of $z$
equals $t_i + (m_i-1)|u|-(i-1)$. Summing over all $i$ and rearranging yields the claimed form of $n$.
Since $n \geq \ell t + (|u|-1)\frac{\ell(\ell-1)}{2}$, we deduce that $\sum_{i=1}^{\ell}(m_i-i) \geq 0$. Therefore, $r:=\sum_{i=1}^\ell m_i \geq \ell(\ell+1)/2$. There thus exist integers $r_\ell > r_{\ell-1} > \cdots > r_1 \geq 1$ such that $r = \sum_{i=1}^{\ell} r_i$.
If we consider the occurrence of $z$ in $u^{n+\ell}$ such that the letter $z_i$ appears in the copy $u_{r_i}$,
we get a contribution of $q^n$ as in \eqref{eq:exponent-form}. This suffices for the claim.
\end{claimproof}

Assume that $[q^n]\qbin{u^{n+\ell}}{z}\neq 0$ with $n$ so large that the above claim holds.
Then there exist tuples $(t_{\ell},\ldots,t_1)$ and $(m_{\ell},\ldots,m_1)$, with the $m_i$ distinct, such that $n$ is as in \eqref{eq:exponent-form}.
Let $m = \sum_{i=1}(m_i-i)$ and let $\alpha_m$ denote the coefficient $[q^m]\qbin{m+\ell}{\ell}$.
Now $\alpha_m$ is the number of index-tuples $(m_{\ell},\ldots,m_{1})$ such that
$\sum_{i=1}^{\ell}m_i-i = m$.
For $\ell$ fixed, we have $\alpha_m \sim \frac{m^{\ell-1}}{\ell!(\ell-1)!} + \mathcal{O}(m^{\ell-2}) = \Omega(n^{\ell-1})$ by \cite[Thm.~2.4]{Stanley2015SomeAR} (since the $t_i$ are bounded and $|u|$ is constant). By \eqref{eq:exponent-form}, the $\alpha_m$ index tuples
$(m_{\ell},\ldots,m_1)$ give distinct occurrences of $z$ in $\qbin{u^{n+\ell}}{z}$, each of which contribute $q^n$ to the $q$-binomial coefficient. We have thus showed that $[q^{n}]\qbin{u^{n+\ell}}{z} \geq \alpha_m = \Omega(n^{k-1})$.

Let us now discuss which term $q^n$ of the series appear with a non-zero coefficient. Consider an arbitrary occurrence of $z$ as a subword of $u^n$ providing a non-zero coefficient for some $q^n$ with associated factorizations as in \eqref{eq:ups-bis}. By reducing \eqref{eq:exponent-form} modulo $|u|$,
we conclude that
\[n\equiv \sum_{j=1}^\ell t_j - \frac{\ell(\ell-1)}{2} \pmod{|u|}\]
for some admissible $\ell$-tuple $(t_1,\ldots,t_\ell)$. If this is not the case, then the corresponding coefficient in the series is vanishing.

As a conclusion, we have thus shown that the coefficient of every large enough power $q^n$ is non-zero if and only if there exists a $\ell$-tuple associated with some factorization of the form \cref{eq:ups-bis} such that $n$ is congruent to $\sum_{j=1}^\ell t_j - \ell(\ell-1)/2$ modulo $|u|$.
\end{proof}

\begin{example}\label{rem:vanish}
In \cref{exa:powunz}, with $z=01$ and $u=0110$, we have two choices for $t_1\in\{1,2\}$ (the positions of $z_1=1$ in $u$) and also two choices for $t_2\in\{0,3\}$ (the positions of $z_2=0$ in $u$). So the pairs $(t_1,t_2)$ are $(1,0)$, $(1,3)$, $(2,0)$ and $(2,3)$. Modulo~$4$, the sum $t_1+t_2-1$  may take the values $0,1,3$ and we see that the coefficient $\langle\mathfrak{s}_{(0110)^\omega,01},q^{2+4n}\rangle=0$ for all $n$.
\end{example}

\section{Extra properties of $q$-Parikh matrices}
In this section we consider some other properties of $q$-Parikh matrices, as well as some inequalities that follow.
\subsection{Properties of minors and other relations}

A {\em minor} of a matrix is the determinant of a square submatrix obtained by removing one (or more) of its rows and columns. The value of each minor of an arbitrary Parikh matrix is a non-negative integer \cite[Thm.~6]{tMaYu02a} and it still holds true for Parikh matrix induced by a word \cite[Cor.~21]{Serbanuta2004}. This is an easy application in linear algebra. We have a similar result in the $q$-deformed case and we provide the proof for the sake of completeness.
\begin{proposition}
  Any minor of $\mathcal{P}_z(u)$ is a polynomial with non-negative integer coefficients.
\end{proposition}

\begin{proof}
  Proceed by induction on the length of $u$. The result trivially holds if $|u|=0,1$. Now assume that the result holds for words of length at most $n$ and consider the word $dw$ of length $n+1$ where $d$ is a letter and $|w|=n$. We have that
  \[\mathcal{P}_z(dw)=\mathcal{M}_{d,|w|}\mathcal{P}_z(w).\]
  Assume that $d$ occurs in $z$ in positions $i_1,\ldots,i_s$. This means that the row of index $i_j$ in $\mathcal{P}_z(dw)$, for $j\in\{1,\ldots,s\}$ is equal to the sum of the $i_j$th row of $\mathcal{P}_z(w)$ and $q^{|w|}$ times the $(i_j+1)$st row of $\mathcal{P}_z(w)$. By linearity of the determinant, any minor of $\mathcal{P}_z(dw)$ can thus be expressed as a linear combination of minors of $\mathcal{P}_z(w)$ with coefficients $1$ or $q^{|w|}$. By induction hypothesis, minors of $\mathcal{P}_z(w)$ are polynomials with non-negative integer coefficients.
\end{proof}

The following corollary is an immediate consequence.

\begin{corollary}
  Let $M=(\mathcal{P}_z(u))^{-1}$ and $1\le i\le j\le |z|+1$, $(-1)^{i+j}M_{i,j}$ is a polynomial with non-negative integer coefficients.
\end{corollary}

Let us focus on $2\times 2$ minors occurring above the main diagonal. They are of the form
\begin{align}
  \begin{split}
\label{eq:minor2}
  \begin{vmatrix}
    q^{\mathsf{s}(|vw|-1)}\qbin{u}{vw} & q^{\mathsf{s}(|vwx|-1)}\qbin{u}{vwx}\\
    q^{\mathsf{s}(|w|-1)}\qbin{u}{w}   & q^{\mathsf{s}(|wx|-1)}\qbin{u}{wx}
\end{vmatrix} =&
q^{\mathsf{s}(|vw|-1)+\mathsf{s}(|wx|-1)}  \qbin{u}{vw}\qbin{u}{wx}\\
& -q^{\mathsf{s}(|w|-1)+\mathsf{s}(|vwx|-1)}\qbin{u}{w}\qbin{u}{vwx}.
\end{split}
\end{align}
for some factors $v,w,x$ of $z$ such that $z=pvwxs$, $p,s\in A^*$.
\begin{remark}
  We observe that in the context of $q$-deformed rational numbers, it is shown in \cite[Thm.~2]{MRO} that if $r(q)/s(q)$ and $r'(q)/s'(q)$ are two $q$-rationals, then $rs'-r's$ is a polynomial in $q$ with positive integer coefficients.
\end{remark}
The fact that \cref{eq:minor2} belongs to $\mathbb{N}[q]$ is the $q$-analogue of what Salomaa calls the {\em Cauchy inequality} \cite{SalomaaCauchy2003}
\[\binom{u}{vw}\binom{u}{wx}\ge\binom{u}{w}\binom{u}{vwx}.\]
It would be interesting to see if such a polynomial \cref{eq:minor2} has a combinatorial interpretation, or corresponds to (the product of) other $q$-deformed binomial coefficients. The following example shows that this is at least not the product of some Gaussian binomial coefficients

\begin{example}
Take $u=ababba$ and $z=bba$, $v=w=b$ and $x=a$. One can verify that the minor associated to those factors is $q^{13}+q^{12}+q^{10}$. Since the power $q^{11}$ is missing, this cannot be a product of Gaussian coefficients (if it were the case, we would have consecutive powers of $q$ in the polynomial because of the unimodality of Gaussian binomials).
\end{example}

In \cite{SalomaaCauchy2003}, a ``dual'' of the Cauchy inequality is considered. In our setting, we get the following.
\begin{proposition}
  For all words $x,y,z,w$, the polynomial
  \[\qbin{xy}{w}\qbin{yz}{w}-\qbin{xyz}{w}\qbin{y}{w}\]
  has non-negative integer coefficients.
\end{proposition}
\begin{proof} We list all pairs of occurrences of $w$ as a subword of both $xyz$ and $y$ and injectively match these pairs of occurrences with distinct occurrences in $xy$ and $yz$. So every contribution to the second term is always compensated by a contribution to the first term. There are three cases.
  
  If we focus on an occurrence of $w$ in $xyz$ where all the selected letters appear in $yz$, then the same occurrence appears in $yz$ and they contribute equally to $\qbin{xyz}{w}$ and respectively $\qbin{yz}{w}$. Similarly, any occurrence of $w$ in $y$ appears in the same position in $xy$. They contribute equally to $\qbin{y}{w}$ and respectively $\qbin{xy}{w}$.

  If we focus on an occurrence of $w$ in $xyz$ where all the selected letters appear in $xy$ (and at least a letter of $x$ is selected, because the case where $w$ is a subword of $y$ has been treated above), then the same occurrence appears in $xy$. The contribution of the first to $\qbin{xyz}{w}$ has an extra factor $q^{|w|\, |z|}$ compared with the contribution of the second to $\qbin{xy}{w}$.   Similarly, any occurrence of $w$ in $y$ appears in $yz$. They contribute to $\qbin{y}{w}$ and respectively $\qbin{yz}{w}$ and the second one has an extra factor $q^{|w|\, |z|}$.

  Finally, the remaining occurrences of $w$ in $xyz$ are such that at least one letter is selected in both $x$ and $z$. We use the same notation as Salomaa. Let $xyz=x_m\cdots x_0y_\ell\cdots y_0z_k\cdots z_0$. Let us consider the occurrence given by \[w=x_{i_{|w|}}\cdots x_{i_{p+q+1}}y_{i_{p+q}}\cdots y_{i_{p+1}}z_{i_p}\cdots z_{i_1}\]
where $p\ge 1$, $p+q+1\le |w|$ and $k\ge i_p> \cdots >i_1\ge 0$, $\ell\ge i_{p+q}> \cdots >i_{p+1}\ge 0$, $m\ge i_{|w|}> \cdots >i_{p+q+1}\ge 0$. We also consider any occurrence $y_{j_{|w|}}\cdots y_{j_1}$ of $w$ within $y$, $\ell\ge j_{|w|}> \cdots >j_{1}\ge 0$. Let $\gamma_{p+n}=\min(i_{p+n},j_{p+n})$ and  $\delta_{p+n}=\max(i_{p+n},j_{p+n})$ for $n=1,\ldots,q$. Now $w$ occurs in $xy$ as \[x_{i_{|w|}}\cdots x_{i_{p+q+1}}y_{\gamma_{p+q+1}} \cdots y_{\gamma_{p+1}} y_{j_p}\cdots y_{j_1}\] and $w$ occurs in $yz$ as
\[y_{j_{|w|}}\cdots y_{j_{p+q+1}} y_{\delta_{p+q+1}} \cdots y_{\delta_{p+1}} z_{i_p}\cdots z_{i_1}.\]
As in the previous case a common power $q^{(|w|-p)}$ appears. To conclude with the proof, the reader may observe that the considered pairs of occurrences of $w$ within $xy$ and $yz$ and pairwise distinct. 
\end{proof}

\subsection{On expressing generalized $q$-Parikh matrices with E{\u{g}}ecio{\u{g}}lu's}
Recall the notation $\mathcal{E}_k$ from \cref{def:Ek} in which we consider the specific word $12\cdots k$. As pointed out by \c{S}erb\u{a}nu\c{t}\u{a} \cite{Serbanuta2004}, elements of a Parikh matrix $\mathcal{P}_z(u)$ associated with a word~$z$ can be related to elements of a classical\footnote{classical in the sense that it is a matrix like those studied initially in \cite{MateescuSalomaa2001}. Note here that this is achieved by adopting a larger alphabet to avoid the redundancies that can appear in $z$.} Parikh matrix $\mathcal{E}_{|z|}(\sigma_z(u))$ for which the author considers a particular morphism that we recall below. Hence algebraic properties of Parikh matrices can be transferred to matrices associated with a word. We investigate this question for our $q$-deformations. 

As a preliminary comment, for an arbitrary non-erasing morphism $\varphi:A^*\to B^*$ and words $w\in A^*$, $u\in B^*$, there is a formula to compute  $\qbin{\varphi(w)}{u}$ as 
\[
\sum_{\ell=1}^{|u|}
\sum_{\substack{ u_1,\ldots,u_\ell\in B^+\\ u=u_1\cdots u_\ell}}
\sum_{\substack{w=w_0a_1w_1\cdots a_\ell w_\ell\\ a_1,\ldots a_\ell\in A\\ w_1,\ldots,w_\ell \in A^*}}
\qbin{\varphi(a_1)}{u_1}\cdots \qbin{\varphi(a_\ell)}{u_\ell}
q^{\sum_{i=1}^\ell |u_i|(|\varphi(w_ia_{i+1}\cdots a_\ell w_\ell )|-|u_{i+1}\cdots u_\ell|)}.\]
In particular, if the morphism is $r$-uniform, then the exponent can be rewritten as
\[\sum_{i=1}^\ell |u_i|(r(\ell-i)+r|w_i\cdots w_\ell|-|u_{i+1}\cdots u_\ell|).\]
This is the $q$-analogue of \cite[Thm.~24]{LejeuneLeroy2020} and it can be deduced form \cref{thm:powers}. The idea is to highlight subwords $u_i$ forming an occurrence of $u$ within blocks of the form $\varphi(a_i)$ for some of the letters $a_i$ constituting $w$. The information about the exponent corresponding to such an occurrence of $u_i$ in $\varphi(a_i)$ is encoded by $\qbin{\varphi(a_i)}{u_i}$. It still needs to be corrected by the number of letters to the right of the block $\varphi(a_i)$, i.e., $|\varphi(w_ia_{i+1}\cdots a_\ell w_\ell )|$ and different from the last letters of $u$ that still needs to be taken into account.

\begin{definition}
Let $z=z_1\cdots z_\ell$ be a word of length~$\ell$ over $A$. For all $a\in A$, we define the morphism $\sigma_z:A^*\to\{1,\ldots,\ell\}^*$ by 
\[\sigma_z(a)=j_1\cdots j_r \quad \text{ whenever } z_{j_1}=\cdots=z_{j_r}=a,\]
i.e., $\sigma_z$ maps a letter to the word encoding the positions of its occurrences within $z$.\end{definition}
As an example, with $z=121323$, we have \[\sigma_z:1\mapsto 13,\ 2\mapsto 25,\quad 3\mapsto 46.\]
Notice that $\sigma_z(z_i)$ always contains $i$. Conversely, every $i\in\{1,\ldots,\ell\}$ appears in exactly one of the images $\sigma_z(j)$, the one such that $z_i=j$.

\begin{lemma}  Let $z$ be a word such that $z_i\neq z_{i+1}$ for all $1\le i<|z|$.
  Let $1\le i\le j\le |z|$, the $q$-binomial
  \[\qbin{\sigma_z(z_i\cdots z_j)}{i\cdots j}\] is a monomial of the form $q^n$ for some $n$. Otherwise stated, $i\cdots j$ appears exactly once in $\sigma_z(z_i\cdots z_j)$.
\end{lemma}
\begin{proof}
By assumption $z_k\neq z_{k+1}$, $i \leq k < j$. Hence the letters $k$ and $k+1$ cannot both appear in the image of a letter. To get an occurrence of $i\cdots j$ in $\sigma_z(z_i\cdots z_j)$, we must have that
the occurrence of $k$ is the occurrence of $k$ in $\sigma_z(z_k)$. Since $k$ appears only once in $\sigma_z(z_k)$ by definition of the morphism $\sigma_z$, the claim follows.


\end{proof}
\begin{proposition}\label{prop:extraProperty}
  Let $A$ be an alphabet of size~$k$. Let $z$ be a word of length $k.r$ such that $|z|_a=r$ for all $a\in A$ and $z_i\neq z_{i+1}$ for all $1\le i<k.r$. We have
  \begin{align}
  \qbin{\sigma_z(u)}{i\cdots j} &= \qbin{\sigma_z(z_i\cdots z_j)}{i\cdots j}\cdot \qbin[q^r]{u}{z_i\cdots z_j}\label{eq:extraProperty1} \quad \text{ and }\\
    \qbin{\sigma_z(u)^{R}}{i\cdots j} &= \qbin{\sigma_z(z_i\cdots z_j)^R}{j \cdots i}\qbin[q^r]{\widetilde{u}}{z_i\cdots z_j} \label{eq:extraProperty2}
  \end{align}
  for all $1\le i\le j\le k.r$. Here we have $x^R = \widetilde{x}$ for the sake of readability.
\end{proposition}

\begin{proof}
We first prove first \eqref{eq:extraProperty1}. There is a one-to-one correspondence between the occurrences of $z_i\cdots z_j$ in $u$ and $i\cdots j$ in $\sigma_z(u)$. Consider one such occurrence and write $u=x_{i-1}z_ix_i\cdots x_{j-1}z_jx_j$ where $x_k$ are words. From \cref{thm:powers}, the contribution of this occurrence to $\qbin{u}{z_i\cdots z_j}$ is
  \[
  q^{\sum_{m=i}^j (m-i+1) |x_m|}.
  \]
  This particular occurrence corresponds to an occurrence of $i\cdots j$ within $\sigma_z(u)$ as factorized below
  \[
  \sigma_z(u)=\sigma_z(x_{i-1})\alpha_i\, i\, \beta_i\sigma_z(x_i)\cdots\sigma_z(x_{j-1})\alpha_j\, j\, \beta_j\sigma_z(x_j)
  \]
  where we have highlighted the occurrence of $k\in\{i,\ldots,j\}$ within $\sigma_z(z_k)=\alpha_k k \beta_k$. The contribution of this occurrence to $\qbin{\sigma_z(u)}{i\cdots j}$ is
  \[
  q^{\sum_{m=i}^j (m-i+1) |\beta_m\sigma_z(x_m)\alpha_{m+1}|}=
  q^{r\sum_{m=i}^j (m-i+1) |x_m|}.q^{\sum_{m=i}^j (m-i+1) |\beta_m\alpha_{m+1}|},
  \]
  where we set $\alpha_{j+1}=\varepsilon$ and we used the fact that $\sigma_z$ is an $r$-uniform morphism. To conclude with the proof, observe that the second factor on the r.h.s.~is the contribution to
 $\qbin{\sigma_z(z_i \cdots z_j)}{i\cdots j}$ of the occurrence of $i\cdots j$ within $\alpha_i \, i\, \beta_i\cdots \alpha_j\, j\, \beta_j$.
  
We then prove \eqref{eq:extraProperty2} . There is a one-to-one correspondence between the occurrences of $z_j\cdots z_i$ in $u$ and $z_i\cdots z_j$ in $\widetilde{u}$ (and occurrences of $i\cdots j$ in $\sigma_z(u)^R$).
Consider one such occurrence and write $u=x_{j}z_j x_{j-1} \cdots x_{i}z_ix_{i-1}$ where the $x_k$ are words.
We thus have
\begin{align*}
\widetilde{u} &= \widetilde{x}_{i-1}z_i\widetilde{x}_{i}\cdots \widetilde{x}_{j-1}z_j \widetilde{x}_{j} \quad \text{and}\\
\sigma_z(u)^R &= \sigma_z(x_{i-1})^R \sigma_z(z_i)^R \sigma_z(x_{i})^R \cdots \sigma_z(x_{j-1})^R \sigma_z(z_j)^R \sigma_z(x_{j})^R.
\end{align*}
From \cref{thm:powers}, the contribution of the occurrence of $z_i\cdots z_j$ to $\qbin{\widetilde{u}}{z_i\cdots z_j}$ is
$q^{\sum_{m=i}^j (m-i+1) |\widetilde{x}_m|}$. For each $k \in \{i,\ldots,j\}$, write $\sigma(z_k)^R = \alpha_k k \beta_k$.
Then the corresponding occurrence of $i\cdots j$ in $\sigma_z(u)^R$ contributes
  \[
  q^{\sum_{m=i}^j (m-i+1) |\beta_m\sigma_z(x_m)^R\alpha_{m+1}|}=
	  q^{r\sum_{m=i}^j (m-i+1) |x_m|}.q^{\sum_{m=i}^j (m-i+1) |\beta_m \alpha_{m+1}|}
  \]
  to $\qbin{\sigma_z(u)^R}{i\cdots j}$. (Again we set $\alpha_{j+1}=\varepsilon$ and we have used the fact that $\sigma_z$ is $r$-uniform.)
  We thus find
  \[
  \qbin{\sigma_z(u)^R}{i\cdots j} = \qbin[q^r]{\widetilde{u}}{z_i\cdots z_j} \cdot q^{\sum_{m=i}^j (m-i+1) |\beta_m \alpha_{m+1}|}.
  \]
  It remains to show that
  \begin{equation}\label{eq:annoyingEquality}
  q^{\sum_{m=i}^j (m-i+1) |\beta_m \alpha_{m+1}|} =  \qbin{\sigma_z(z_i \cdots z_j)^R}{j \cdots i}.
  \end{equation}
  Notice that $\sigma_z(z_i \cdots z_j)^R = \alpha_j j \beta_j \cdots \alpha_i i \beta_i$,
  so that the right-hand-side of \cref{eq:annoyingEquality} is 
  \[
  q^{\sum_{m = i}^j (j-m+1)|\beta_m\alpha_{m-1}|}.
  \]
Here we set $\alpha_{i-1} = \varepsilon$.

On the other hand, we observe that the term on the left-hand-side of \cref{eq:annoyingEquality} is also equal to 
\begin{align*}
  \qbin{(\sigma_z((z_i \cdots z_j)^R))^R}{i\cdots j} &= q^{(r-1)(j+1-i)^2}\qbin[1/q]{\sigma_z((z_i \cdots z_j)^R)}{j \cdots i}\\
  &=
  q^{(r-1)(j+1-i)^2 - \sum_{m=i}^j (j-m+1)|\widetilde{\alpha}_{m}\widetilde{\beta}_{m-1}| },
  \end{align*}
  where we let $\beta_{i-1} = \varepsilon$ and the first equality comes from \cref{cor:reversal}.
  Now to show that \eqref{eq:annoyingEquality} is true, it is enough to show that
  \[
  \sum_{m=i}^j (j-m+1)|\widetilde{\alpha}_{m}\widetilde{\beta}_{m-1}|  + \sum_{m = i}^j (j-m+1)|\beta_m\alpha_{m-1}| = (r-1)(j+1-i)^2.
  \]
  Indeed, joining the sums on the left-hand-side, we simplify to
\begin{align*}   
\sum_{m=i}^j (j-m+1)|\alpha_m\beta_m\alpha_{m-1}\beta_{m-1}| &= (r-1)(j-i+1) + 2(r-1)\sum_{m=i+1}^j(j-m+1)\\
&= (r-1)(j-i+1) + 2(r-1)\sum_{k=1}^{j-i}k \\&= (r-1)(j - i+1 + (j-i)(j-i+1)) = (r-1)(j+1-i)^2.
  \end{align*}
  Thus the proof is complete.
\end{proof}

In the next statement if $M$ is a matrix whose entries are polynomials in $q$, an expression of the form $M(q^r)$ means that we have substituted $q$ by $q^r$ in every entry. 

\begin{corollary}\label{cor:new_expression}
Let $z$ be as in \cref{prop:extraProperty}. Let further $Z$ (resp. $C$) be the upper triangular matrix whose above-diagonal entries $Z_{i,j+1}$ (resp., $C_{i,j+1}$), $i \leq j \leq |z|+1$, are of the form
  $\qbin{\sigma_z(z_i\cdots z_j)}{i\cdots j}$ (resp. $q^{\mathsf{s}(j-i)}$).
  Then
  \[
  C(q^{r-1})\odot \mathcal{E}_{|z|}(\sigma_z(u)) = Z \odot \mathcal{P}_z(u)(q^r)
  \]
  and
   \[
   C(q^{r-1}) \odot \mathcal{E}_{|z|}(\sigma_z(u))^{-1} = Z \odot (\mathcal{P}_z(u)(q^r))^{-1}.
   \]
   \end{corollary}
   \begin{proof}
     We inspect the element at position $i$, $j+1$, the right-hand-side of the first equality is
     \[\qbin{\sigma_z(z_i\cdots z_j)}{i\cdots j} q^{r\cdot \mathsf{s}(j-i)} \qbin[q^r]{u}{z_i\cdots z_j}
     =q^{(r-1)\cdot \mathsf{s}(j-i)} q^{ \mathsf{s}(j-i)} \qbin{\sigma_z(u)}{i\cdots j}\]
     where we used \cref{eq:extraProperty1}.
  
For the second equality, inspecting the element at position $i$, $j+1$,
and using \cref{thm:inverse2},
the equality is equivalent to
\begin{multline*}
q^{(r-1)\cdot \mathsf{s}(j-i)}\cdot (-1)^{i+j+1}q^{(j+1-i)(r|u|-1)}q^{-\mathsf{s}(j-i)}\qbin[1/q]{\sigma_z(u)^R}{i \cdots j}\\
=
\qbin{\sigma_z(z_i\cdots z_j)}{i\cdots j}(-1)^{i+j+1}q^{r(j+1-i)(|u|-1)} q^{-r\cdot \mathsf{s}(j-i)} \qbin[1/q^r]{\widetilde{u}}{z_i \cdots z_j}.
\end{multline*}
Rearranging gives
\begin{align*}
\qbin[1/q]{\sigma_z(u)^R}{i \cdots j} &= q^{-(r-1)(j+1-i)^2}\qbin{\sigma_z(z_i\cdots z_j)}{i\cdots j}\qbin[1/q^r]{\widetilde{u}}{z_i\cdots z_j}\\
	&= \qbin[1/q]{(\sigma_z(z_i\cdots z_j))^R}{j \cdots i}\cdot \qbin[1/q^r]{\widetilde{u}}{z_i\cdots z_j}
\end{align*}
by \cref{cor:reversal}.
This is equivalent to \eqref{eq:extraProperty2} by replacing $1/q$ with $q$ in the argument.
\end{proof}

With $z=121323$, $u=1121323$, we have $\sigma_z(z)=132513462546$ and $\sigma_z(u) = 02021402351435$.
Verifying the above theorem, we have $r=2$,
  \[
  Z=\left(\begin{smallmatrix}
    1 & q & q^3 & q^5 & q^9 & q^{13} & q^{18}\\
    0 & 1 & q & q^2 & q^5 & q^8 & q^{12}\\
    0 & 0 & 1 & 1 & q^2 & q^4 & q^7\\
    0 & 0 & 0 & 1 & q & q^2 & q^4\\
    0 & 0 & 0 & 0 & 1 & 1 & q\\
    0 & 0 & 0 & 0 & 0 & 1 & 1\\
   0 & 0 &0 & 0 & 0 & 0 & 1 \\\end{smallmatrix}\right),
   \quad 
   \text{and}
   \quad
   C=\left(\begin{smallmatrix}
1 & 1 & q & q^{3} & q^{6} & q^{10} & q^{15} \\
0 & 1 & 1 & q & q^{3} & q^{6} & q^{10} \\
0 & 0 & 1 & 1 & q & q^{3} & q^{6} \\
0 & 0 & 0 & 1 & 1 & q & q^{3} \\
0 & 0 & 0 & 0 & 1 & 1 & q \\
0 & 0 & 0 & 0 & 0 & 1 & 1 \\
0 & 0 & 0 & 0 & 0 & 0 & 1
\end{smallmatrix}\right).
   \]
Further,
\[\mathcal{P}_z(u)=
\left(\begin{smallmatrix}
1 & q^{6} + q^{5} + q^{3} & q^{10} + q^{9} + q^{7} + q^{6} + q^{4} & q^{13} + q^{12} & q^{15} + q^{14} + q^{13} + q^{12} & q^{16} + q^{15} & q^{16} + q^{15} \\
0 & 1 & q^{4} + q & q^{7} & q^{9} + q^{7} & q^{10} & q^{10} \\
0 & 0 & 1 & q^{6} + q^{5} + q^{3} & q^{8} + q^{7} + q^{6} + 2q^{5} + q^{3} & q^{9} + q^{8} + q^{6} & q^{9} + q^{8} + q^{6} \\
0 & 0 & 0 & 1 & q^{2} + 1 & q^{3} & q^{3} \\
0 & 0 & 0 & 0 & 1 & q^{4} + q & q^{6} + q^{4} + q \\
0 & 0 & 0 & 0 & 0 & 1 & q^{2} + 1 \\
0 & 0 & 0 & 0 & 0 & 0 & 1
\end{smallmatrix}\right),
  \]
and finally
$\mathcal{E}_6(\sigma_z(u))$ is computed as
\[ \left(\!\!\begin{smallmatrix}
1 & q^{13} + q^{11} + q^{7} & q^{22} + q^{20} + q^{16} + q^{14} + q^{10} & q^{28} + q^{26} & q^{33} + q^{31} + q^{29} + q^{27} & q^{35} + q^{33} & q^{35} + q^{33} \\
0 & 1 & q^{9} + q^{3} & q^{15} & q^{20} + q^{16} & q^{22} & q^{22} \\
0 & 0 & 1 & q^{12} + q^{10} + q^{6} & q^{17} + q^{15} + q^{13} + 2q^{11} + q^{7} & q^{19} + q^{17} + q^{13} & q^{19} + q^{17} + q^{13} \\
0 & 0 & 0 & 1 & q^{5} + q & q^{7} & q^{7} \\
0 & 0 & 0 & 0 & 1 & q^{8} + q^{2} & q^{12} + q^{8} + q^{2} \\
0 & 0 & 0 & 0 & 0 & 1 & q^{4} + 1 \\
0 & 0 & 0 & 0 & 0 & 0 & 1
\end{smallmatrix}\!\!\!\!\right).
  \]
 One can then verify that $C(q)\odot \mathcal{E}_6(\sigma_z(u)) = Z\odot \mathcal{P}_z(u)(q^2)$.

\bibliographystyle{plainurl}
\bibliography{./bibliography}

\end{document}